\documentclass{article}
\usepackage{amssymb}
\usepackage{graphicx}
\usepackage{amsmath}
\usepackage{hyperref}
\usepackage[toc,page]{appendix}
\newtheorem{theorem}{Theorem}

\newtheorem{corollary}[theorem]{Corollary}

\newtheorem{definition}[theorem]{Definition}

\newtheorem{lemma}[theorem]{Lemma}

\newtheorem{proposition}[theorem]{Proposition}
\newtheorem{remark}[theorem]{Remark}

\newenvironment{proof}[1][Proof]{\textbf{#1.} }{\ \rule{0.5em}{0.5em}}

\DeclareMathAlphabet{\mathpzc}{OT1}{pzc}{m}{it}
\newcommand {\Gr}{\mbox{{\rm Gr}}}

\newcommand {\id}{\mbox{{\rm id}}}

\renewcommand {\L}{{\cal L}}

\newcommand {\W}{{\cal W}}

\newcommand {\F}{\mathcal{F}}

\renewcommand {\O}{\mathcal{O}}
\newcommand {\OO}{\mathbb{O}}

\newcommand{\susy}{\mathfrak{susy}}

\renewcommand{\Im}{\mbox{{\rm Im}}}
\newcommand{\Tor}{\mbox{{\rm Tor}}}

\newcommand{\Sym}{\mbox{{\rm Sym}}}
\newcommand {\Ker}{\mathrm{ Ker}}
\newcommand{\GL}{\mathrm{GL}}
\newcommand{\Spin}{\mathrm{Spin}}
\newcommand{\SO}{\mathrm{SO}}
\newcommand{\g}{\mathfrak{g}}
\newcommand{\h}{\mathfrak{h}}
\newcommand{\dbar}{\bar \partial}

\newcommand{\E}{\mathcal{E}}
\newcommand{\so}{\mathfrak{so}}
\newcommand{\I}{\mathcal{I}}

\newcommand{\Ad}{\mathrm{Ad}}
\newcommand{\OGr}{\mathrm{OGr}}
\newcommand{\Proj}{\mathrm{Proj}}
\newcommand{\p}{\mathfrak{p}}
\newcommand{\uu}{\mathfrak{u}}

\newcommand{\Bl}{\mathrm{Bl}}
\newcommand{\gl}{\mathfrak{gl}}
\renewcommand{\sl}{\mathfrak{sl}}

\newcommand{\wX}{\widetilde{X}}
\renewcommand{\P}{\mathbf{P}}
\newcommand{\SL}{\mathrm{SL}}
\newcommand{\ind}{\mathrm{ind}}

\newcommand{\ddef}{\mathrm{def}}

\renewcommand{\t}{\mathfrak{t}}
\newcommand{\sd}{\partial}
\newcommand{\OFl}{\mathrm{OFl}}

\newcommand{\RRe}{\mathrm{Re}}
\newcommand{\Hom}{\mathrm{Hom}}
\newcommand{\calR}{\mathcal{R}}
\newcommand{\rad}{\mathfrak{rad}}
\newcommand{\V}{\mathcal{V}}
\newcommand{\A}{\mathcal{A}}
\newcommand{\m}{\mathfrak{m}}
\newcommand{\Ext}{\mathrm{Ext}}
\newcommand{\cExt}{\mathcal{EXT}}
\newcommand{\cGr}{\mathcal{GR}}
\newcommand{\rank}{\mathrm{rank}}
\newcommand{\Pic}{\mathrm{Pic}}

\newcommand{\Pf}{\mathrm{Pf}}
\renewcommand{\H}{\mathcal{H}}
\newcommand{\vv}{V}
\newcommand{\<}{\langle}
\renewcommand{\>}{\rangle}
\newcommand{\Cech}{\check{\mathrm{C}}\mathrm{ech}}
\begin{document}
\title{Geometry of a desingularization of eleven-dimensional gravitational spinors.} 
\author{M.V. Movshev\\Stony Brook University\\Stony Brook, NY, 11794-3651,\\ USA \\ \texttt{mmovshev@math.sunysb.edu}} 

\maketitle

\begin{abstract}
We show that the space of gravitational spinors in eleven dimensions, defined by equations $\Gamma_{\alpha\beta}^i\lambda^{\alpha}\lambda^{\beta}=0$ admits a desingularization with nice geometric properties. In particular  the desingularization   fibers  over the isotropic Grassmannian $\OGr(2,11)$. This enables us to recast  equations of linearized eleven-dimensional supergravity adapted to 3-form potential   into Cauchy-Riemann equations on a super extension of  isotropic  Grassmannian  $\OGr(2,11)$. 
\end{abstract}

{\bf Mathematics Subject Classification(2010).} \\
main  83E50, secondary 14M17,14M30,18G15 \\
{\bf Keywords.} Supergravity, pure spinors, fibration

\tableofcontents
\section{Introduction}
One of the mathematical problems in eleven-dimensional supergravity  (\cite{ECremmerBJuliaandJScherk}) is to find a  formulation of the theory where  symmetries, including supersymmetries, have a geometric interpretation. The classical superspace formulation (\cite{BrinkHowe}, \cite{CremmerFerrara}) makes supersymmetries manifest, with  a drawback that the fields it encodes  are not unconstrained but satisfy supergravity equations. Proposal  \cite{Cederwall} is supposed to rectify this. Cederwall's construction still depends on the choice of a  background solution of supergravity equations, but the fields are unconstraint. In  the flat background the fields in his theory  are elements of an algebra 
\begin{equation}\label{E:gr}
Gr^{\infty}=A\otimes \Lambda[s^*_{11}]\otimes C^{\infty}(\mathbb{R}^{11}).
\end{equation}

 Commutative algebra $A$ is defined to be 
\begin{equation}\label{E:Adef}
A=\mathbb{C}[\lambda^{1},\dots,\lambda^{32}]/(v^i)
\end{equation}
\begin{equation}\label{E:pure}
v^i=\Gamma^{i}_{\alpha\beta}\lambda^{\alpha}\lambda^{\beta}=0, i=1,\dots, 11, \alpha,\beta=1,\dots,32
\end{equation}
where $\Gamma^{i}_{\alpha\beta}$ are eleven-dimensional $\Gamma$-matrices (see e.g. \cite{Deligne}  for mathematical introduction). Variables $\lambda^{\alpha}$ are linear coordinates on the spinor representation $s_{11}$   of the complex spinor group $\Spin(11)$. The generators $\theta^{1}, \dots,\theta^{32} $ of the Grassmann algebra $\Lambda[s^*_{11}]$ are linear functions on the odd spinor representation $s_{11}$.

This algebra $Gr^{\infty}$  is equipped with  the differential
\[D=\lambda^{\alpha}\frac{\sd}{\sd \theta^{\alpha}}-\Gamma^i_{\alpha\beta}\lambda^{\alpha}\theta^{\beta}\frac{\sd}{\sd x^i},\]

According to \cite{CGNN},\cite{CNT},\cite{MemBerkovits} cohomology of $Gr^{\infty}$ coincide with the space of solutions of the linearized equations of eleven-dimensional supergravity.

We find the name of eleven-dimensional pure spinors used for $X$ in the physics literature (\cite{BerNek}, \cite{Cederwall}) mathematically  misleading, because it is already reserved for another closely related object. Not having a better alternative we propose to call $X$ the space of gravitational spinors.

 In order to take advantage of analytic methods  (cf. \cite{Cederwall} where  Dolbeault forms has been used to write the Lagrangian) it is reasonable to desingularise $X$.  
  The space $X$ contains a subvariety $X_{sing}$  defined by equations 
\begin{equation}\label{E:matrix}
v^{ij}= \Gamma^{i_1i_2}_{\alpha\beta}\lambda^{\alpha}\lambda^{\beta}=0, 1\leq i_1,i_2\leq 11,
\end{equation} as a singular locus.
 A natural way to   desingularise $X$ is to blow it up  along $X_{sing}$. This desingularization $\wX$  has already been used in  \cite{BerNek} for a computation of  virtual characters. 
 
 Our prime goal  is to carry out a mathematical study of topology and algebraic geometry of $\wX$.

  Submanifold $X_{sing}$ is isomorphic to the space of Cartan's pure spinors $\OGr(5,11)$ (\cite{Cartan}). The fundamental representation $\vv^{11}$ of the complex orthogonal group $\SO(11)$ is equipped with invariant complex-linear inner product $(\cdot,\cdot)$. The space of pure spinors is a one in the series of isotropic (also called orthogonal) Grassmanians \[\OGr(k,11)=\{F\subset \vv^{11}|\dim F=k, (\cdot,\cdot)|_{F}=0 \}.\]

 Let us briefly go  over the characteristic features of $X$ and $\wX$. By results of Igusa \cite{Igusa}    $X$ partitions into a union of two $\Spin(11)$ orbits $O_{15}\cong \OGr(5,11)$ and $O_{22}$ of dimensions 15 and 22 respectively. It turns out that $\wX$     is smooth and possesses a fibration 
  \begin{equation}\label{E:projection}
 p:\wX \rightarrow \OGr(2,11).
 \end{equation}

 \begin{equation}\label{E:projectionmain}
 \lambda\overset{p}{\rightarrow }\Gamma^{i_1i_2}_{\alpha\beta}\lambda^{\alpha}\lambda^{\beta}e_{i_1}\wedge e_{i_2}
 \end{equation}
A  fiber $p^{-1}(x)$ is isomorphic to a projective space $\P^7$. 
The exceptional divisor $Y\subset \wX$ of the blowdown
 \[\Bl:\wX\rightarrow X\] is a total space of a fibration on quadrics
 \begin{equation}\label{p:res}
 Y\rightarrow \OGr(2,11),
 \end{equation}
  $p^{-1}(x)\cap Y\cong Q\subset \P^7$.
 Variety $Y$ is isomorphic to the space of partial isotropic flags 
 \[\OFl(2,5,11)=\{F_1\subset F_2\subset \vv^{11}|\dim F_1=2,\dim F_2=5, (\cdot,\cdot)|_{F_2}=0 \}.\] Projection 
 \[\OFl(2,5,11)\rightarrow \OGr(2,11)\] coincides with (\ref{p:res}); the map
  \[\OFl(2,5,11)\rightarrow \OGr(5,11)\] is the blowdown onto $O_{15}$.
  The diagram 
\begin{equation}\label{E:dualitydiagram}
X\overset{\Bl}{\leftarrow}\wX\overset{p}{\rightarrow}\OGr(2,11)
\end{equation} can be thought of as an eleven-dimensional analogue of duality diagrams $X\leftarrow Y\rightarrow Z$ studied in \cite{Calib}. Penrose-Radon transform (\cite{Calib}) can be employed to carry  geometric objects from $X$ to $Z$. This idea  can be adapted to  eleven-dimensional supergravity setup advocated in  \cite{CGNN},\cite{CNT},\cite{MemBerkovits}. Here is the precise statement. 

The algebra $Gr^{\infty}$ contains a subalgebra $Gr^{pol}=A\otimes \Lambda[s^*_{11}]\otimes \Sym[\vv^{11}]$. $\Sym$ stands for symmetric or polynomial algebra; in the presence of the inner product $(\cdot,\cdot)$ we make no distinction between $\vv^{11}$ and $(\vv^{11})^*$. $Gr^{pol}$ is an object associated with the left space in the duality diagram (\ref{E:dualitydiagram}).

Let $P=P_2$ be the  stabilizer $St(x)\subset \Spin(11), x\in \OGr(2,11)$.  
The spinor representation $s_{11}$ has a filtration $F_1\subset F_2$ by $P$-invariant subspaces. Dimension of $F_1$ is eight, dimension of $F_2$ is twenty four. The linear space $\t=F_1$ is an abelian subalgebra in the supersymmetry algebra $\susy=\vv^{11}+\Pi s_{11}$ (see Section \ref{S:homodef} for details). We use $\Pi$ for the parity change operation. We define $L$ to be a algebraic homogenous superspace of super-Poicar\'e group $\Spin(11)\ltimes \susy$ with isotropy subgroup $P\ltimes \Pi \t$. $L$ is an object associated with the right space  in the diagram (\ref{E:dualitydiagram}). Eleven-dimensional manifestation of Penrose-Radon transform takes the following form.

\begin{proposition}
There is an isomorphism of the cohomology of $Gr^{pol}$
and the cohomology of the structure sheaf of $L$. 
\end{proposition}
Note that cohomology in the analytic version of the theory can be computed through $\dbar$ complex. This makes a connection between  cohomology of $Gr^{an}=A\otimes \Lambda[s^*_{11}]\otimes \O^{an}(\vv^{11})$ an Dolbeault cohomology of the structure sheaf of $L$.

Roughly the proof  consists of two steps. In the first step we replace the tensor factor  $A$ in $Gr^{pol}$ by a quasiisomorphic differential graded algebra $(B,d_B)$.  The $\Lambda[s^*_{11}]$ and  $\Sym[\mathbb{C}^{11}]$ stay unchanged. The differential  $D$ changes to $d_B+D$. To be more precise $B$ is related to $A$ by a sequence of quasiisomorphisms $A\cong B_1\cong B_2\cong B_3=B$.  To define $B_1$ we replace the graded algebra $A=\bigoplus_{i\geq 0} A_i$ by the \v{C}ech complex $B_1=\bigoplus_{i\geq 0}\Cech^{\bullet}(\O(i))$ of some affine covering of $X$. 
%In Appendix (Proposition \label{P:vanishing}) we show that cohomology of $\Cech$ coincide with $A$. 
Algebra $B_2$ is the \v{C}ech complex $\bigoplus_{i\geq 0}\Cech^{\bullet}(\Bl^*\O(i))$ of the pullback of sheaves $\O(i)$ on $\wX$. 
%The pullback in our case does not affect the cohomology (Proposition \ref{E:blowvanish}). 
A fiber of the projection $p$ (\ref{E:projection}) is a projective space $\P^7$. The space of the global sections of $\Bl^*\O(i)$ over $\P^7$ is a space of homogeneous polynomials in eight variables $\lambda^1,\dots,\lambda^8$ of  degree $i$. Globalizing it we get  a polynomial-valued  sheaf $p_*\Bl^*\O(i)$ over the base of $p$ fibration. The algebra $B_3=B$ is the \v{C}ech complex $\bigoplus_{i\geq 0}\Cech^{\bullet}(p_*\Bl^*\O(i))$ over $\OGr(2,11)$. The core of the second step is  an observation that $\lambda^1,\dots,\lambda^8$ form a contracting pair with some eight generators $\theta^1,\dots,\theta^{8}$ in $\Lambda[s^*_{11}]$.

Several interesting problems remained out of the scope of this note.
\begin{enumerate}
\item  Lagrangian of the linearized theory should have a form $\int_L f\dbar f d\mu$, where $d\mu$ is some integral volume form on $L$.     We have not attempted to find a formula for it.

\item The work \cite{Cederwall} gives a description of supergravity Lagrangian $\L_{\mathrm{SUGRA}}$ in a superspace formulation with auxiliary gravitational spinor fields (former eleven dimensional pure spinors). Some of the terms of $\L_{\mathrm{SUGRA}}$ can be interpreted as objects defined  on $\wX$ or  on $\OGr(2,11)$. It is tempting to speculate that the Lagrangian can be defined  on $\wX$ or even on  $L$.   
\item The space $\vv^{11}\times \Pi s_{11}$ is equipped with a non-integrable odd distribution, defined by differential forms $d\theta^{\alpha}-\Gamma_{\alpha\beta}^i\theta^{\beta}dx^i$.  The space $L$ is a moduli of complex purely odd  eight-dimensional  integrable subspaces.  An  example of  a point in this moduli is given by  $\Pi \t$. A  $(1|8)$-dimensional object appeared in \cite{Nilsson}, \cite{witten} in a  description of ten-dimensional supergravity and Yang-Mills theory. We think this analogy worth a more close investigation. 
\end{enumerate}
We plan to address these questions in the following publications.

Here is an outline of the paper. In Section \ref{S:local} we undertake a detailed analysis of the local structure of $X$, by exhibiting  affine charts. Octonion calculus turns to be indispensable. In Section \ref{S:prop}   we establish most of algebro-geometric facts  claimed in this introduction. In Section \ref{S:applications} we discuss the gravitational applications. We moved some technical computation into appendix.  In particular in Appendix \ref{S:homology} we settle some  questions in homological algebra related to gravitational spinors and in particular prove that computations of cohomology in \cite{CNT} pertinent to eleven-dimensional supergravity are mathematically correct. For this we  heavily use algebra system {\it Macauly2}.

All algebraic groups, varieties and linear spaces in this note are defined over the field of complex numbers if not stated otherwise. We also used convention of summation over repeated indices.

The author would like to thank Jason Starr for his help with algebraic geometry and A.S. Schwarz for  stimulating conversations. The concluding part of the work on this paper has been done at IHES. The author would like to thank this institution for the hospitality and inspiring mathematical environment.

\section{Local structure of $X$}\label{S:local}
Local structure of $X$ is far from obvious, because $X$ is not smooth. It is neither obvious that $X$ is reduced. We settle these issues  in this section.

We start with a reminder of the minimal set of facts about $\Gamma$-matrices.

Symmetric tensor square $\Sym^2 s_{11}$ decomposes (see Appendix in  \cite{BVinbergALOnishchik}) into the direct sum of representations
\begin{equation}\label{E:symdecomposition}
\Sym^2 s_{11} \cong \vv^{11}+\Lambda^2\vv^{11}+\Lambda^5\vv^{11}
\end{equation}

 In particular there is a unique linear $\Spin(10)$-equivariant $\Gamma$ map 
 \begin{equation}\label{E:gammast}
 \Sym^2s_{11}\rightarrow \vv^{11},
 \end{equation} which coefficients  in a basis in $s_{11}$ and an orthonormal basis in $\vv^{11}$ are $\Gamma_{\alpha\beta}^i,\alpha=1,\dots 32, i=1,\dots 11 $. We shall use also an intertwiner $\Sym^2s_{11}\rightarrow \Lambda^{2}\vv^{11}$, which gives rise to $\Gamma_{\alpha\beta}^{ij}$.

The algebraic variety $X$ is defined as the  projective spectrum of $A$ (\ref{E:Adef}).

Following \cite{Igusa} we define a polynomial  function $J(\lambda)$ on $s_{11}$ by the formula
\[J(\lambda)=\Gamma_{i\alpha\beta}\Gamma^i_{\gamma\delta}\lambda^{\alpha}\lambda^{\beta}\lambda^{\gamma}\lambda^{\delta}\]
This function is manifestly $\Spin(11)$ invariant. Let $N$ be  projectivization of subset of spinors, satisfying $J(\lambda)=0$.
Igusa \cite{Igusa} found that $N$ is a union of four orbits of dimensions $15,22,24,30$. In addition $\P(s_{11})$ contains  one open orbit $J(\lambda)\neq 0$ of dimension $31$.

It is a well known fact of  algebraic geometry (or better of  general topology) that orbits of connected algebraic groups are irreducible in Zariski topology. 

\begin{remark}\label{E:uniquness}Theory of highest weights (\cite{FultonRep}) implies that  the projective space of an irreducible 
representation of a semi-simple group contains a unique closed orbit - the orbit of the highest weight vector.
\end{remark} 
It means that the orbit $O_{15}$ is closed, but $O_{22}$ of dimension 22 is  not; the complement  $\overline{O}_{22}\backslash O_{22}$ must coincide with $O_{15}$ (other orbits  do not intersect with the closure by dimensional reasons). The closure of an irreducible set is irreducible (see e.g. \cite{Borel}), and we arrive to the following proposition
\begin{proposition}
 The closure $\overline{O}_{22}$ is  irreducible.
\end{proposition} 
\begin{corollary}
Reduced scheme $X_{red}$ coincides with $\overline{O}_{22}$.
\end{corollary}
\begin{proof}
Dimension of reduced, closed in $\P^{31}$ $\Spin(11)$-invariant scheme $X_{red}$ is $22$ (Corollary \ref{C:degree}). Then it must be a union of $O_{15}$ and $O_{22}$.
\end{proof}

The remaining part of this section contains a proof of the statement that 
$X=X_{red}$. This is shown by explicit construction of affine charts and the corresponding local rings.

\subsection{Affine charts of the  smooth locus.}
To construct an affine chart in $X$ containing a point $\lambda\in O_{22}$ we shall need to develop some representation theory.

We start with describing accurately the embedding  $St(\lambda)\subset \Spin(11)$.
According to \cite{Igusa} the  Levi factor of  $St(\lambda)$ is $G_2\times \GL(2)$. In order to characterize its embedding into $\SO(11)$ 
we fix an orthogonal decomposition: 

 \begin{equation}\label{E:decomposition1}
\vv^{11}\cong \vv^7+U+ U^{'}=\vv^7+\vv^4
\end{equation}

In the following $\vv^i$  stands for i-dimensional complex Euclidean space. Two-dimensional spaces   $U,U^{'}\subset \vv^4$ are isotropic and $U\cap U^{'}=0$. The inner product defines a non-degenerate pairing between $U$ and $U^{'}$.
The Lie group $G_2$ has a defining 7-dimensional orthogonal representation in $\vv^7$ (see e.g. \cite{Bryant}).

 Decomposition of $\Ad\SO(11)\cong\Lambda^2\vv^{11}$ into $G_2\times \GL(2)$-irreducible components  takes a form of a grading by weights of a central element $c\in \gl_2$:

\begin{equation}\label{E:decomp2}
\begin{split}
&\Ad(\mathfrak{so}_{11})_2=\Lambda^2U\\
&\Ad(\mathfrak{so}_{11})_1=\vv^7\otimes U\\
&\Ad(\mathfrak{so}_{11})_0=\Lambda^2\vv^7+U\otimes U^{'}=\vv^7+\Ad G_2 +\Ad \GL(2)\\
&\Ad(\mathfrak{so}_{11})_{-1}=\vv^7\otimes  U^{'}\\
&\Ad(\mathfrak{so}_{11})_{-2}=\Lambda^2  U^{'}
\end{split}
\end{equation} 
We have taken into account an isomorphism $\Lambda^2(V)_{G_2}\cong \Ad(G_2)+V$ (see \cite{Bryant} for more on $G_2$ representations). The Lie algebra of the unipotent radical of $St(\lambda)$ (\cite{Igusa}) is 
\begin{equation}\label{E:radical}
\rad=\Ad(\mathfrak{so}_{11})_{-1}+\Ad(\mathfrak{so}_{11})_{-2}.
\end{equation} 

The space of (Dirac) spinors $s_{11}$ in eleven dimensions can be constructed as irreducible module over a Clifford algebra $Cl(\vv^{11})$. Decomposition (\ref{E:decomposition1}) gives rise to an isomorphism
\[Cl(\vv^{11})\cong Cl(\vv^7)\otimes Cl(\vv^4)\]
of $\mathbb{Z}_2$-graded algebras. It in turn gives an identification of irreducible $Cl$-representations \[s_{11}\cong s_{7}\otimes s_4, \dim(s_{11})=32, \dim(s_{7})=8,\dim(s_{4})=4 \]
It is well known (\cite{Deligne}) that $s_{2n+1}$ is an irreducible representation of $\Spin(2n+1)$, whereas $s_{2n}$ is a sum of two chiral (Weil) representation $S_{2n}+S'_{2n}$. 

The complex group $\Spin(4)$ is isomorphic to $\SL(2)\times \SL(2)$. Let $W_l$ and $W_r$ be defining representations of left and right copies of $\SL(2)$. Then $s_4$ is isomorphic to $W_l+W_r$.
We arrived to  an isomorphism of $\Spin(7)\times \Spin(4)$ representations
\begin{equation}\label{E:spinordec}
s_{11}=s_7\otimes W_l+s_7\otimes W_r
\end{equation}
In the spin notations  $\vv^4$ is   isomorphic to $W_l\otimes W_r$.

\subsubsection{Spinor representation of $\Spin(7)$}
In order to solve the system \ref{E:pure} in a neighborhood of solution $\lambda\in O_{22}$ we need to make description of $s_7$ more explicit. We use construction of $s_7$  that utilizes  Cayley numbers $\OO$. 

More precisely we make an identification $s_{7}=\OO\otimes \mathbb{C}$.
To describe the relevant $\Gamma$-matrices and the structure of the Clifford module on $s_7$  we recall some basic properties of $\mathbb{R}$-algebra $\OO$. Following \cite{ConwaySmith} we denote non-associative multiplication in $\OO$ by $x.y$. In $\OO$ an alternating property holds: any subalgebra generated by two elements is associative. In those cases where the product of elements is associative we shall   drop the dot sign.

We shall use some of the properties of octonions. Among these are
\begin{enumerate}
\item  Existence of  anti-involution $x\rightarrow \bar{x}$, $\overline{xy}=\bar{y}\bar{x}$.
\item Decomposition into real and imaginary parts: $x=\frac{x+\bar{x}}{2}+\frac{x-\bar{x}}{2}=\RRe(x)+\Im(x)$
\item Positivity of the inner product  $(x,y)=\RRe(x\bar{y})$. 
\item  $x\bar{x}=||x||^2\in\mathbb{R} \subset \OO$. This implies a  formula for the inverse: $x^{-1}=\frac{\bar{x}}{||x||^2}$  
\item Identity $x^{-1}xy=y$  holds in the non-associative $\OO$. In particular if $x\in \Im(\OO)\overset{\ddef}{=}V_{\mathbb{R}}$ then the formula $xxy=-||x||^2y$ defines an action of the Clifford algebra $Cl(V_{\mathbb{R}})$.
\end{enumerate}
We define a skew-symmetric  map $\gamma:\Lambda_{\mathbb{R}}^2\OO\rightarrow V_{\mathbb{R}}$ by the formula
\[\gamma(x,y)=\Im(x\bar{y})=x\bar{y}-\RRe(x\bar{y})=x\bar{y}-(x,y)\] 

According to \cite{ConwaySmith} 
$(vy,x) = (v,x\bar{y})$, which implies that for  $v\in V_{\mathbb{R}}$ the operator of multiplication on $v$ is adjoint to $\gamma$.

If we set $\vv^7=V_{\mathbb{R}}{\otimes} \mathbb{C}$ then complexification of $\gamma$  by general theory of Clifford modules \cite{Chevalley}   gives a  $\Spin(7)$-map $\Lambda^2s_7\rightarrow \vv^7$.

\subsubsection{Description of elven dimensional $\Gamma$ matrices in terms of $\gamma$}
In this section we make use of the information about $s_7$ found in the previous section to write eleven-dimensional $\Gamma$ matrices in more elementary terms of decomposition (\ref{E:spinordec}).

We start with an observation that $\Gamma$-maps analogous to (\ref{E:gammast}) exist in all dimensions \cite{Deligne}.

The $\Gamma$ map in four dimensions is the projection 
\[\Sym^2(W_l+ W_r)\cong \Sym^2W_l+W_l\otimes W_r+\Sym^2W_r\rightarrow W_l\otimes W_r\cong \vv^4\]
We  equip  linear spaces $W_l$ and $W_r$  with symplectic $\SL(2)$-invariant dot products $\omega_l$, $\omega_r$.
Details about spinors in four dimensions   can be found in \cite{PenroseRindler}.

The $\Gamma$-map in eleven dimensions can be formulated in terms of four and seven-dimensional $\Gamma$-matrices and a choice of $\omega_l$, $\omega_r$. 

Under identification  
\begin{equation}\label{E:gamma1}
\begin{split}&\Sym^2[s_7\otimes W_l+s_7\otimes W_r]\cong \Sym^2[s_7\otimes W_l]+s_7\otimes W_l\otimes s_7\otimes W_r+\Sym^2[s_7\otimes W_r]
\end{split}
\end{equation}
the map $\Gamma$ is a sum of intertwiners:
\[\Sym^2[s_7\otimes W_i]\rightarrow \vv^7, i\in \{l,r\}\]
\begin{equation}\label{E:goct1}
\theta\otimes w\otimes\theta'\otimes w'  \rightarrow \gamma(\theta,\theta')\omega_i(w\wedge w')
\end{equation}
\[s_7\otimes W_l\otimes s_7\otimes W_r\rightarrow W_l\otimes W_r\cong \vv^4\]
 \begin{equation}\label{E:goct2}
 \theta\otimes w_l\otimes\theta'\otimes w_r\rightarrow (\theta,\theta')w_l\otimes w_r
 \end{equation}
 The mathematical proof of this statement follows from decomposition  
 \begin{equation} \label{E:spin7dec}
 \Lambda^2s_7\overset{\gamma(2)\oplus\gamma}{\longrightarrow}\Lambda^2V+V \quad
  \Sym^2s_7\overset{q\oplus \gamma(3)}\longrightarrow  \mathbb{C}+\Lambda^3V
  \end{equation}
  (see  \cite{BVinbergALOnishchik}), surjectivity of $\Gamma$ (\ref{E:gammast}) and Schur lemma.

\subsubsection{$G_2$-equivariant solution of gravitational spinor constraint}\label{S:g2solution}
In this section we solve the gravitational spinor constraint (\ref{E:pure}) in the neighborhood of a generic point. The language of octonions turns to be very useful. Appearance of octonions, which symmetry group is a compact form of $G_2$, is not surprising. The same group as we know is is a part of $St(\lambda)$ for generic $\lambda\in X$.

 To adapt decomposition (\ref{E:spinordec}) to our needs we note that spinor representation $s_7$  upon restriction on $G_2\subset \Spin(7)$ splits into the sum of the  the defining representation $\vv^7$ and the trivial representation \cite{Bryant}.
A $G_2\times \Spin(4)$-equivariant identification

\begin{equation}\label{E:spinorG2}
s_{11}=\vv^7\otimes W_l +W_l+\vv^7\otimes W_r +W_r.\
\end{equation} 

will be used to  solve equations  (\ref{E:pure}).

We choose a basis $e_i,i=1,2$ in $W_l$ and $f_i,i=1,2$ in $W_r$ such that $\omega_l(e_1,e_2)=\omega_r(f_1,f_2)=1$.

A spinor $\lambda\in s_{11}$ can be decomposed into the sum 
\begin{equation}\label{E:spinorGcoordinates}
\lambda=v^1\otimes e_1+v^2\otimes e_2+u^1\otimes f_1+u^2\otimes f_2+w^1e_1+w^2\otimes e_2+r^1f_1+r^2\otimes f_2 
\end{equation}
for some $v^1,v^2,u^1,u^2\in \vv^7$, $w^i,r^i\in \mathbb{C}$. We can say that $\lambda$ is a four-spinor with (complexified) octonion coefficients. 
Equations (\ref{E:pure}) becomes
\begin{equation}\label{E:defeq}
\begin{split}
&-v^1 v^2+w^2v^1- w^1v^2- (v^1,v^2)  -u^1u^2+r^2u^1- r^1u^2- (u^1,u^2) =0\\
&(v^1,u^1)+w^1r^1=0\quad (v^1,u^2)+w^1r^2=0
\end{split}
\end{equation}
\begin{equation}\label{E:defeqlast}
\begin{split}
&(v^2,u^1)+w^2r^1=0 \quad (v^2,u^2)+w^2r^2=0
\end{split}
\end{equation}

We plan to solve these equation for $v^2$ and $r^1,r^2$.
To do this we consider an operator $A:\OO\rightarrow \OO$ defined by the formula
\[A_{l}(x)=lx+(l,x)\] $l\in \Im\OO$
\begin{lemma}
\begin{enumerate}
\item The operator leaves $\Im \OO$ invariant. 
\item The inverse is equal to 
\[A^{-1}_{l}(x)=l^{-1}(x-\left(\frac{l}{\RRe(l)},x\right))\]
\end{enumerate}
\end{lemma}
\begin{proof}
\begin{enumerate}
\item Suppose $\RRe x=0$, i.e $\bar{x}=-x$. Then $(A_{l}(x),1)=(lx+(l,x),1)=(lx,1)+(l,x)=(l,x)+(l,x)=0$.
\item
\[\begin{split}
&A^{-1}_lA_l(x)=l^{-1}\left(lx+(l,x)-\left(\frac{l}{\RRe(l)},lx+(l,x)\right)\right)=\\
&x+(l,x)l^{-1}-\left(\frac{l}{\RRe(l)},lx\right)l^{-1}-(l,x)\frac{(l,1)}{\RRe(l)}l^{-1}=x
\end{split}\]
We use an identity in $\OO$:
\begin{equation}\label{E:oidentity1}
(la,lb)=(\bar{l}la,b)=(l,l)(a,b)
\end{equation} (see \cite{ConwaySmith})
 \end{enumerate}
\end{proof}

The first of the equations (\ref{E:defeq}) can be written as 
\[v^1 v^2+w^1v^2+ (v^1,v^2)  =w^2v^1-u^1u^2+r^2u^1- r^1u^2- (u^1,u^2)\]
The left hand side can be interpreted as $A_l(v^2)$ ($v^i,u^i\in \OO\otimes \mathbb{C}$), where $l=w^1+v^1$. The formula for $A_l^{-1}$ enables us to solve equation with respect to $v^2$. Simultaneously we eliminate $r^1,r^2$ using the second pair of equations (\ref{E:defeq}):

\begin{equation}\label{E:vrformulas}
\begin{split}
&v^2=\frac{w^2v^1}{w^1}+\frac{1}{(v^1,v^1)+w^1w^1}\left(-w^1u^1u^2-(v^1,u^2)u^1+(v^1,u^1)u^2-w^1(u^1,u^2)\right.\\
&\left. +v^1.(u^1.u^2)+\frac{(v^1,u^2)}{w^1}v^1u^1-\frac{(v^1,u^1)}{w^1}v^1u^2+(u^1,u^2)v^1- \frac{(v^1,u^1u^2)}{w^1}v^1+(v^1,u^1u^2) \right)\\
&\left.+\frac{w^2(v^1,v^1)}{w^1}v^1-\frac{(v^1,u^1u^2)}{w^1}v^1-w^2(v^1,v^1)+(v^1,u^1u^2) \right)\\
&r^1=-(v^1,u^1)/w^1\quad r^2=-(v^1,u^2)/w^1\\
\end{split}
\end{equation}

\begin{proposition}\label{P:subs}
Upon substitution $v^2,r^1,r^2$ (\ref{E:vrformulas})

equations (\ref{E:defeqlast}) become  identities.
\end{proposition}
\begin{proof} See Appendix \ref{A:subs}.
\end{proof}

Let $U$ be an open subset in $X$ defined by inequalities $(v^1,v^1)+w^1w^1\neq0, w^1\neq 0$. It contains a spinor $\lambda$ with all vanishing projective coordinates but one  $w^1=1$.
\begin{corollary} \label{L:locsmooth}
 The  scheme $(U,\O_U)$ is reduced. The algebra  $\O_U$ is a localization of a polynomial algebra
 $(S)^{-1}\mathbb{C}[w^1,w^2]\otimes \Sym [s^*_7+s^*_7+s^*_7]$ , where a multiplicative set $S$ is generated by $\{w^1,(v^1,v^1)+w^1w^1\}$.
\end{corollary}

Scheme $X\backslash O_{15}$ is reduced because $\Spin(11)$ acts transitivity on $(X\backslash O_{15})_{red}=O_{22}$.
 In any reduced scheme smooth points form an open subset (see e.g. \cite{Humphreys} p.40).
 
 Thus $X\backslash O_{15}$ is a smooth scheme,  because of the transitive $\Spin(11)$ action.  Our conclusion is that scheme-theoretically $X\backslash O_{15}=O_{22}$
 
\begin{corollary}\label{E:canform}
By the  reason of transitivity of $\Spin(11)$-action, for any $\lambda\in O_{22}$ there is a subgroup $G_2\times \GL(2)\in St(\lambda)$ such that in the corresponding decomposition (\ref{E:spinorG2}) 
$\lambda$ has all coordinates in (\ref{E:spinorGcoordinates}) but $w^1$  equal to zero.
\end{corollary}

\subsection{Affine charts near the singular locus}\label{S:singchart}
In this section we exhibit an affine chart in $X$ near a point $\lambda\in O_{15}$. According to \cite{Igusa} the Levi factor of $St(\lambda)$ coincides with $\gl_{5}$. Following the method  of the previous section we shall decompose $s_{11}$ into $\gl_5$ irreducibles. A spinor $\lambda$ will be accordingly broken into  the sum $\sum_i\lambda_i$ of  components. Equations (\ref{E:pure}) will be solved for some of $\lambda_i$.
\subsubsection{$\gl_5$ decomposition of $s_{11}$}\label{S:decompgl}
Our plan  is to  explicitly describe local rings near  $X_{sing}$. Our analysis follows closely  \cite{BerNek}.

Spinor representation $s_{11}$ has a description in terms of Grassmann algebra \cite{Chevalley}.  Let $A$ and $A^{'}$ be a five-dimensional isotropic subspace in $\vv^{11}$ such that $A\cap A^{'}=0$. We have a decomposition
 \begin{equation}\label{E:adecomp}
 \vv^{11}=\vv^{10}+\vv^1=A+A^{'}+\vv^1
 \end{equation}
where  $A+A^{'}=\vv^{10}$ is a 10-dimensional Euclidean space. A one dimensional  subspace  $\vv^1$ spanned by vector $u$ is orthogonal to $\vv^{10}$.  The bilinear form $(\cdot,\cdot)$ defines a pairing between $A$ and $A^{'}$. The group $\GL(5)$ acts  on $A+A^{'}$ preserving $(\cdot,\cdot)$ and trivially on $\vv^1$.

 This  defines an embedding 
 \begin{equation}\label{E:embedding}
 \GL(5)\subset \SO(10)\subset \SO(11).
 \end{equation} Spinor representation $s_{11}$ restricted on double cover $\widetilde{\GL}(5)$ is isomorphic to 
 $\Lambda(A^{'})\otimes \det^{\frac{1}{2}}$ (see \cite{Cartan}).  We shall  drop $\det^{\frac{1}{2}}$-factor in the following to simplify notations.

Spinor representation $s_{11}$ is symplectic \cite{Deligne}. Let $\Omega$ be the corresponding skew-symmetric $\Spin(11)$-invariant inner product.     
Components of   the $\Gamma$-maps are $\Omega$-adjoint to multiplication 
  $A^{'}\otimes \Lambda^iA^{'}\rightarrow \Lambda^{i+1}A^{'}$,  contraction  $A\otimes \Lambda^{i+1}A^{'}\rightarrow \Lambda^{i}A^{'}$ and the Clifford multiplication defined by $u$: $u|_{\Lambda^{i}A^{'}}=(-1)^i\id$.  

The spaces  $ \Lambda^{i}A^{'}$ and $ \Lambda^{5-i}A^{'}$ are $\Omega$-dual. The only nontrivial components  of $\Gamma$-map are:
\[\begin{split}&\Gamma^{+}:\Lambda^{i}A^{'}\otimes \Lambda^{4-i}A^{'}\rightarrow A\\
&\Gamma^{-}:\Lambda^{i}A^{'}\otimes \Lambda^{6-i}A^{'}\rightarrow A^{'}\\
&\Gamma^0:\Lambda^{i}A^{'}\otimes \Lambda^{5-i}A^{'}\rightarrow \mathbb{C}\end{split}\]

 Let $e^i$ be a basis in $A^{'}$. Elements $\{e_{i_1}\wedge\dots \wedge e_{i_k}|i_1<\cdots<i_k\}$ form a basis in $\Lambda^k(A^{'})$. A spinor $\lambda$ can be written as
 \[\lambda=\sum_{k=0}^5\lambda_i=\sum_{k=0}^5u^{i_1,\dots i_k}e_{i_1}\wedge\dots \wedge e_{i_k}\]
where $u^{i_1,\dots i_k}$ are variables defined for $i_1< \cdots <i_k$. We  set them to zero for other combinations of indices.

 Let $\iota_r:\Lambda^iA^{'}\rightarrow \Lambda^{i-1}A^{'}$ a graded-differentiation, defined as a contraction with a vector $r\in A$. In the spirit of \cite{BerNek} equations (\ref{E:pure}) can be written in terms of components $\lambda_i$:
 \begin{equation}
 \begin{split}
 &\lambda_0\lambda_5+\lambda_1\lambda_4-\lambda_2\lambda_3=0\\
 &-2\lambda_0\lambda_4-2\lambda_1\lambda_3+\lambda_2\lambda_2=0\\
 &-2\iota_r(\lambda_1)\lambda_5+2\iota_r(\lambda_2)\lambda_4+\iota_r(\lambda_3)\lambda_3=0
 \end{split}
 \end{equation}
The third equation should be  valid for every vector $r\in A$.

 The element $\lambda_0$ is a scalar. 
 We solve the first two equation  for $\lambda_4$ and $\lambda_5$.
 Upon a substitution 
 \begin{equation}\label{E:reduced0}
 \lambda_3=\tau_3+\lambda_1\lambda_2/\lambda_0, \tau_3\in \Lambda^3A^{'}
 \end{equation} the third equation becomes 
 \begin{equation}\label{E:reduced}
 \iota_r(\tau_3)\tau_3=0.
 \end{equation}

 We use $\widetilde{\GL}(5)$-invariant inner-product $\Omega$  to identify $\Lambda^3A^{'}$ with $\Lambda^2A$.
 Equation (\ref{E:reduced}) has a simple interpretation in the dual variable :
 \begin{equation}\label{E:plucker}
\tau^2=u_{ij}e^i\wedge e^j\in  \Lambda^2A \quad  \tau^2\tau^2=0
 \end{equation}
$Pol=\mathbb{C}[u_{ij}]/(u_{ij}+u_{ji}), 1\leq i, j\leq 5$ is a polynomial algebra on coefficients of skew-symmetric matrices. Equation $\tau^2\tau^2=0$ can be formulated as relations between $u_{ij}$. These  relations define an ideal in  $Pol$, which is according to (\cite{Fulton}, Section 8.4) is prime.  Relations    describe the Pl\"{u}cker embedding of $\Gr(2,5)$ into $\P^9$ (\cite{Fulton}). 
The affine cone $C\Gr(2,5)$ is singular precisely at the a apex, which is characterized by equation 
\begin{equation}\label{E:sing}
\tau_3=0
\end{equation}
Presented arguments show validity of the following proposition.
 \begin{proposition}\label{P:chart}
The algebra of regular functions on the  affine chart in $X$, defined by equation $\lambda_0\neq0$ is a reduced, irreducible affine scheme.
 The algebra of regular functions is isomorphic to $\mathbb{C}[u,u^{-1},u^{ij}]\otimes B$, $0\leq i<j\leq 5$. 
 \begin{equation}\label{E:bdef}
 B=Pol/(\Pf(M_k)), k=1,\dots,5
 \end{equation}
 $M_k$ is a sub-matrix of $(u_{ij})$ obtained by removing $k$-th row and $k$-th column. $\Pf$ is the Pfaffian. 
$B$ is the algebra of homogenous functions on $\Gr(2,5)\subset \P^{9}$.
 \end{proposition}
 We formulate  decomposition of   the adjoint  $\SO(11)$-representation under $\GL(5)\subset \SO(11)$ as a $\mathbb{Z}$ grading. The graded index is the eigenvalue of a central element $c\in \gl_5$.
  \begin{equation}\label{E:decomp5}\begin{split}
  &\Ad(\SO(11))_2=\Lambda^2 A^{'}\\
 &\Ad(\SO(11))_1=A^{'}\\
 &\Ad(\SO(11))_0=\Ad(\GL(5))\\
 &\Ad(\SO(11))_{-1}=A\\
 &\Ad(\SO(11))_{-2}=\Lambda^2 A
 \end{split}\end{equation}
 The parabolic Lie subalgebra $\p_5=\gl_5\ltimes \uu''$ ($\uu''=\Ad(\SO(11))_{-1}+\Ad(\SO(11))_{-2}$ ) exponentiates to an algebraic Lie group 
\begin{equation}\label{E:parabolic}
P_5\subset \Spin(11)
\end{equation} 

Elements of $A^{'}$ and $\Lambda^2 A^{'}$ act in $s_{11}$ by multiplication, elements of $A$ and $\Lambda^2A$ by contraction. In particular the unit vector $1\in \Lambda(A^{'})$ is invariant with respect to $\p_5$. This makes it a lowest vector.

 The $\Spin(11)$-orbit of $\lambda=1$ in $\P(s_{11})$ coincides  \cite{Cartan} with:

\begin{equation}\label{E:isoorbits}
O_{15}\cong \Spin(11)/P_5\cong \OGr(5,11)
\end{equation}

 By transitivity of $\Spin(11)$ action on $O_{15}$ we conclude that  points $\lambda\in O_{15}\subset X$ have isomorphic  neighborhoods. 
 The space $X$ can be covered by charts described in Sections \ref{S:g2solution}, \ref{S:decompgl}. This leads to the following corollary
 \begin{corollary}\label{C:reduced}
The scheme $X$ is reduced. 
\end{corollary}

\begin{proposition}\label{P:sing}
The subscheme of singular points $X_{sing}$ is reduced and irreducible. It coincides with the orbit $O_{15}$. 
\end{proposition}
\begin{proof}
In a local chart, described in Proposition \ref{P:chart} $X_{sing}$ is described by equation 
 (\ref{E:sing}). We can express  homogeneous coordinates $\lambda_3,\lambda_4,\lambda_5$ as functions of $\lambda_0,\lambda_1$ and $\lambda_2$:
$\lambda_3=\frac{\lambda_1\lambda_2}{\lambda_0},\lambda_4=\frac{\lambda^2_2}{\lambda_0},\lambda_5=\frac{\lambda_1\lambda^2_2}{2\lambda^2_0} $.
 
We immediately see that the local ring is a subalgebra in the field of fractions $\mathbb{C}( u,u^{i},u^{ij}) $ and has no zero divisors.

The proper closed (see \cite{Hartshorne}) subscheme $X_{sing}$ is invariant under $\Spin(11)$ action.

 By Igusa  classification it must coincide with $O_{15}$.
\end{proof}

\section{Properties of the blowup of $X$}\label{S:prop}
Presently  we shall establish some of the elementary properties of the blowup of $X$.

 We start with a reminder of some classical facts. Let $V_{\mu}$ be an irreducible representation of a semisimple group $G$ of a highest (or a lowest) weight $\mu$. Suppose the highest weight vector space $v_{\mu}$ is invariant with respect to parabolic subgroup $P$.   The formula $g\rightarrow gv_{\mu}$ defines a projective embedding  $G/P\rightarrow \P(V_{\mu})$. A classical theorem of Kostant (see e.g. \cite{LT}, Theorem 1.1) asserts that algebra $R=R_{G/P}=\bigoplus_{i\geq 0} H^0(G/P,\O(i))$  of homogeneous  functions is quadratic, i.e. it is generated by its first graded component, and the ideal of relations is generated by elements of degree two. The algebra contains no zero divisors.  Borel-Weil-Bott theory (see e.g. \cite{Calib}) tells us that the spaces of homogeneous elements $ H(G/P,\O(i))$ are  irreducible representations of $G$.
 
 The closed $\Spin(11)$-orbit $O_{15}\cong \Spin(11)/P_5$ in $\P^{31}$ is an orbit of a lowest vector. 
  The theory described in the previous paragraph allows us to unambiguously  determine relations in $R_{O_{15}}$. We immediately conclude that  $H^0(O_{15},\O(2))$ is isomorphic to  $\Lambda^5\vv^{11}$ (this is a representation from decomposition \ref{E:symdecomposition}); $\Gamma^{ij}_{\alpha\beta}\lambda^{\alpha}\lambda^{\beta}\in \Lambda^2\vv^{11}$ must be  equal to zero on $O_{15}$. 
  Let $CY$ be the affine cone over projective variety $Y$. We denote $p_{aff}$ the map (\ref{E:projectionmain}) of affine cone $CX$ to $\Lambda^2 \vv^{11}$.
  \begin{lemma}\label{L:nontriv}
  For every  $\lambda\in CO_{22}$ $p_{aff}(\lambda)\neq 0$
  \end{lemma} 
  \begin{proof}The group $\Spin(11)$ acts transitively on $CO_{22}\backslash 0$. In the case of failure of the statement $p_{aff}(\lambda)$ would be zero for all $\lambda\in CO_{22}\backslash 0$.
The scheme  $O_{22}$ is reduced (Corollary \ref{C:reduced}), therefore $\Gamma^{ij}_{\alpha\beta}\lambda^{\alpha}\lambda^{\beta}=0$ in algebra $A$.  This contradicts with the definition of  ideal  (\ref{E:Adef}).
  \end{proof}
  
\begin{corollary}\label{C:indeterm}
The  map $p$ (\ref{E:projectionmain}) is well defined on $O_{22}$. Its singular locus is a  subscheme   $O_{15}$.
\end{corollary}

 \begin{definition}\label{D:dense}We define $\widetilde{X}$ to be the closure of the graph $p$ in the product $X\times \P^{54}\subset \P^{31}\times \P^{54}$. For  blowup construction see the end of this  section.    \end{definition}

In the remaining part of this section we shall exhibit a structure of projective fiber-bundle on the blowup $\wX$ over  $\OGr(2,11)$ . 

We start with a construction of a large number of projective subspaces in $X$.
\begin{proposition}\label{E:vanish}
An element $\lambda\in s_{11}$ of the form $v \otimes e \in s_7\otimes W_l$  or $u \otimes f \in s_7\otimes W_r$ (see decomposition \ref{E:spinordec}) satisfies $\Gamma(\lambda,\lambda)=0$.
\end{proposition}
\begin{proof} Follows from formula (\ref{E:defeq}).

\end{proof}

The elements $u\otimes f$ with a fixed $f$ generate  $\P^{7}\subset X$.
One of the  results of \cite{Igusa} is that $\P(s_7)$ is a union of two $\Spin(7)$ orbits. One closed orbit coincides with the quadric $Q$, defined by the $\Spin(7)$-invariant inner product. The other orbit is \begin{equation}\label{E:openorb7}
O_7=\P^{7}\backslash Q.
\end{equation}  The group $St(v), v\in O_7$ is isomorphic to $G_2$. From this we  deduce that as long as $(v,v)\neq0$ we can always find $g\in \Spin(7)$ and transform $v$ into $gv\in \mathbb{C}\subset \mathbb{C}+\mathbb{C}\otimes \Im\OO$

Let $P_2\subset \Spin(11)$ be a connected subgroup, which Lie algebra 
\begin{equation}\label{E:pstab}
\p_2=\Ad(\mathfrak{so}_{11})_0+\Ad(\mathfrak{so}_{11})_{-1}+\Ad(\mathfrak{so}_{11})_{-2}
\end{equation} in the notations of decomposition (\ref{E:decomp2}). The Levi factor of $\p_2$ is  $\so_7\times \gl_2$; the radical $\rad$ is the same as (\ref{E:radical}).

 \begin{proposition}\label{P:invariance}
Subvariety $\P^{7}\subset X$ is invariant with respect to  $P_2$. The action of the unipotent radical in $P_2$ on $\P^7$ is trivial.
\end{proposition}
\begin{proof}
There is a $\mathbb{Z}$-grading on $s_{11}$, that can be obtained in full analogy with   (\ref{E:decomp2}) as a  grading by eigenspaces of the central element $c\in \gl_2$.

\begin{equation}\label{E:grading}
\begin{split}
&s^{1}=s_7\otimes f^+\\
& s^0=s_7\otimes W_l\\
&s^{-1}=s_7\otimes f^-
\end{split}
\end{equation}
where $f^{\pm}$ is a $c$ eigenbasis in $W_r$

Element $s_7\otimes f^-$ has $c$-weight $-1$. Both gradings (\ref{E:decomp2}) and (\ref{E:grading}) are  compatible. Homogenous elements in  $\rad$ have  weights $-1$ and $-2$. Therefore $\rad$ acts on $s_7\otimes f^-$ by zero. We conclude that subspace  $\P(s_7\otimes f^{-})\subset X$ is invariant under $P_2$.
Let $W_r$ be the same as in decomposition (\ref{E:spinordec}).  A choice of a basis $f^{\pm}\in W_r$ uniquely determines a subalgebra $\gl_2\subset\so_4$ in decomposition (\ref{E:decomp2}), for which $f^{\pm}$ are the weight vectors.
\end{proof}

\begin{proposition}\label{P:map1}
The image of the map $p:O_{22}\rightarrow \P^{54}$  (\ref{E:projectionmain}) contains in $\OGr(2,11)\subset \P^{54}$.
\end{proposition}
\begin{proof}

If  we compare decomposition of $\Sym^2[s_{11}]$ (\ref{E:gamma1}) and decomposition of 
\begin{equation}\label{decomp3}
\begin{split}&\Lambda^2 \vv^{11}=\Lambda^2[\vv^7+\vv^4]=\Lambda^2[\vv^7+W_l\otimes W_r]\\
&=\Lambda^2\vv^7+\Sym^2W_l+\Sym^2W_r+V\otimes W_l\otimes W_r, \end{split}
\end{equation}

we see that by representation theoretic reasons there are  only the following nontrivial intertwiners between the components (see formula (\ref{E:spin7dec}) for definition of $\gamma(2)$)
\begin{equation}
\begin{split} 
&\Sym^2[s_7\otimes W_l]\rightarrow \Lambda^2V+\Sym^2W_l\\
& v^1\otimes e^1\otimes v^2\otimes e^2\rightarrow \gamma(2)(v^1,v^2)\omega_r(e^1,e^2)+ (v^1,v^2)e^1e^2 \\
&\\
&s_7\otimes W_l\otimes s_7\otimes W_r\rightarrow V\otimes W_l\otimes W_r\\
& v^1\otimes e^1 \otimes u^2\otimes f^2\rightarrow (v^1,u^2)e^1\otimes u^2\\
&\\
&\Sym^2(s_7\otimes W_r)\rightarrow \Lambda^2V+\Sym^2W_r\\
& u^1\otimes f^1\otimes u^2\otimes f^2\rightarrow \gamma(2)(u^1,u^2)\omega_r(f^1,f^2)+ (u^1,u^2)f^1f^2
\end{split}
\end{equation}
$e^1e^2$ and $f^1f^2$ stand for symmetric product of vectors.

By  Corollary \ref{E:canform} we can assume that   $\lambda\in O_{22}$ has a form $u \otimes f$, $(u,u)\neq 0$.

Then $p_{aff}(\lambda)\neq 0$  (Proposition \ref{L:nontriv})  and $p_{aff}(\lambda)$ is equal to 
\begin{equation}\label{E:quadric}
q(u)ff=(u,u)ff\in \Sym^2 W_r\subset \Lambda^2\vv^{11}
\end{equation} (the $\Lambda^2V$-component is equal to zero because $\omega_r(f,f)=0$). 
We need to show that $ff\in \Lambda^2\vv^{11}$ is a decomposable tensor and it defines an isotropic two-plane. We shall use representation theoretic arguments to show that $ff$ can be identified with an element of $\Lambda^2 U^{'}$  in decomposition (\ref{E:decomp2}). By definition $U^{'}$ is isotropic and proposition would follow. 

It suffice to show that there is subalgebra $\so_7\times \gl_2\subset \so_{11}$ such that $ff$ is invariant under $\so_7$. In addition $ff$  has the same weight under the action of the central part of $\gl_2$ as the elements of $\Lambda^2 U^{'}$.

Note that decomposition (\ref{E:spinordec}) dictates the choice of $\so_7$ and a choice of $\so_4=\sl_2\times\sl_2$ that contains our $\gl_2=\sl_2+<c>$.  In the notations of  decomposition (\ref{decomp3}) the Lie algebra $\sl_2\times\sl_2$ coincides with $\Sym^2W_l+\Sym^2W_r$. We choose $\sl_2\subset \gl_2$ to be equal to $\Sym^2W_l$.  We complete $f$ to an eigenbasis $f,f'$ of $W_r$ for some element $c\in \sl_2=\Sym^2W_r$ The element $c$ is a generator of a Cartan subalgebra in $\sl_2\cong \Sym^2W_r$; elements $f$, $ff$ are the lowest vectors of the corresponding representations. From this we conclude that $ff$ has the same weight as elements of $\Lambda^2 U^{'}$ in decomposition (\ref{E:spinordec}) constructed with respect to our subgroup $\so_7\times \gl_2\subset \so_{11}$.
\end{proof}

 We note in passing  that $ff$ is a nilpotent element in $\so_{11}$.
\begin{corollary}
Variety $\wX$ contains in $X\times \OGr(2,11)\subset X\times \P^{54}$. Projection on $\OGr(2,11)$ defines a $\Spin(11)$-equivariant map  (\ref{E:projection}).
\end{corollary}
By construction the map $p$ (\ref{E:projection}) is {\it projective} (\cite{Hartshorne} p.103).  
The space $\wX$ coincides with the blowup of $X$ associated with the sheaf of ideals generated by sections (\ref{E:matrix}) ( \cite{Hartshorne} Example 7.17.3.).  By  Theorem 4.9 \cite{Hartshorne} the map $p$ is proper.

\subsection{The structure of the map $p:\wX\rightarrow \OGr(2,11)$}\label{S:mapstructure}
Our plan for now is to show that $\wX$ is a total space of a projectivization of an eight-dimensional algebraic vector bundle over $\OGr(2,11)$.

The map (\ref{E:projection}) is a $\Spin(11)$-morphism; therefore the stabilizer $St(x)\subset \Spin(11)$  of  $x\in \OGr(2,11)$ acts on the fiber $p^{-1}(x)$. 
\begin{lemma}
Let $I$ be ideal in $A$ generated by $v^{ij}$.
The blowup algebra $\Bl A=\bigoplus_{n\geq 0} I^n$ is generated by $66$ elements.
In particular $\wX$ is a noetherian scheme of  finite type.
\end{lemma}
\begin{proof}
The algebra $\Bl A$  is generated by $\lambda^{\alpha}$ and $v^{ij}$.
\end{proof}

From this we immediately deduce that $p$ is a morphism of finite type (\cite{Hartshorne}, p.84)
\begin{proposition}(\cite{Mumford}, page 57) Let $X\overset{f}\rightarrow Y$ be a morphism of finite type of noetherian
schemes, and let $\F$ be a coherent sheaf on $X$. Assume that $Y$ is reduced
and irreducible. Then there is a non-empty open subset $U\subset Y$ such that
the restriction of $\F$ to $f^{-1}(U)$ is flat over $U$.
\end{proposition}
Projective homogenous variety $\OGr(2,11)$ is of finite type. We choose $\F$ to be $\O_{\wX}$. The $\Spin(11)$-action and  the previous proposition lets to deduce the following.
\begin{corollary}\label{C:flat}
The map  $p:\wX\rightarrow \OGr(2,11)$ is flat.
\end{corollary}

\begin{proposition}\label{P:fiber}
Fibers of the map $p$ are smooth and isomorphic to $\P^7$.
\end{proposition}
\begin{proof}
We denote $\Bl$-preimage of $O_{22}$ in $\wX$ by $\Delta$. 
Let $r$ be projection $\Delta\rightarrow \OGr(2,11)$. Then the closure $\overline{r^{-1}(x)}$ of  $r^{-1}(x)$ in $\wX$ contains in $p^{-1}(x)$. Let us prove that 
\begin{equation}\label{E:closure}
\overline{r^{-1}(x)}=p^{-1}(x)
\end{equation}

 We shall shaw first that any point $y\in \wX$ belongs to the closure $\overline{r^{-1}(x)}$ for some $x$ (by closure we mean here a closure in analytic topology). Let $a_i\in \Delta$ be a sequence that converges to $y\in \wX$. A compact group $\Spin(11,\mathbb{R})\subset \Spin(11)$ acts transitively on $\OGr(2,11)$ \cite{Montgomery}. We choose $g_i\in \Spin(11,\mathbb{R})$ such that $g_i(p(a_i))$ is equal to some fixed $x'$. Compactness allows us to extract a convergent subsequence $g'_i$ such that $\lim_{i\rightarrow \infty} g'_i=g$. We set $g^{-1}x'=x$. Then $g^{-1}g'_ia_i\in r^{-1}{g^{-1}x'}=r^{-1}(x)$ and $\lim_{i\rightarrow \infty}g^{-1}g'_ia_i=y\in\overline{r^{-1}(x)}$. 

The fiber  $r^{-1}(x)\subset \P^{31}\times \P^{54}$ being $St(x)$-orbit  is locally closed in Zariski topology. By Chow theorem \cite{Chow} its set-theoretic (analytic) closure coincides with the closure in  Zariski topology. This verifies (\ref{E:closure}).

 We shall shaw next that there is an isomorphism  of the fiber $p^{-1}(x)$ onto $\P^7$.
 Projective $\Bl$ maps $p^{-1}(x)$ onto some closed scheme $G_x\subset X$. Schemes $p^{-1}(x)$ and $G_x$ are acted upon by $St(x)\cong P_2$.  By Proposition \ref{P:invariance}  $\Bl\ r^{-1}(x)\subset X$ is a 
  dense  in  $\P^{7}\subset X$. We conclude that $\P^7\subset G_x$.
  The intersection of the locus of indeterminacy of $p:X\rightarrow \OGr(2,11)$ with $\P^7$, defined in the proof of Proposition \ref{P:map1}, coincides with  the  zero set  of  a quadratic function $q(v)=(v,v)$ (\ref{E:quadric}),
  which  defines a sheaf of invertible ideals $I_Q$ on $\P^7$.
 The pre-image of $\P^7$ in $\wX$ under $\Bl$ is $\P^7$. This is because a blowup of a scheme along an invertible sheaf of ideals is a trivial operation \cite{Hartshorne}. 
We have a chain of inclusions $r^{-1}(x)\subset \P^7\subset p^{-1}(x)$, which implies a scheme isomorphism $p^{-1}(x)=\overline{r^{-1}(x)}\cong \P^7$.
\end{proof}
\begin{proposition}\label{P:smooth}(\cite{Hartshorne}, Ex III.10.2)
Let $f:X\rightarrow Y$ be a proper, flat morphism of varieties over $\mathbb{C}$. Suppose for some
point $y\in  Y$ that the fiber $X_y$ is smooth over $\mathbb{C}(y)$. Then there is an
open neighborhood $U$ of $y$ in $Y$ such that $f:f^{-1}(U) \rightarrow U$ is smooth.
\end{proposition}
  \begin{corollary}
 The map $p$ is smooth and defined a structure of projective bundle over $\OGr(2,11)$. In particular $\wX$ is smooth.
  \end{corollary}
 \begin{proof}
The first statement follows directly from Corollary \ref{C:flat} and  Propositions \ref{P:fiber}, \ref{P:smooth}.

 Note that $\wX$ carries the  pullback of the tautological line bundle $\Bl^*\O(1)$. It restriction on $\P^7\cong p^{-1}(x)$ is a degree one line bundle. The direct image 
 \begin{equation}\label{E:vectprb}
 \E=p_*\Bl^*\O(1)
 \end{equation} is an eight-dimensional homogeneous vector bundle on $\OGr(2,11)$, which projectivization is $\wX$.
 \end{proof}
\begin{corollary}\label{C:total}
The exceptional divisor   $Y=\wX\backslash O_{22}$ of the blowdown $\Bl$

is  a fiber bundle of quadrics  over $\OGr(2,11)$.
\end{corollary}
\begin{proof}
We want prove that $p^{-1}(x)\cap Y$ is a quadric. As before  $\Delta$ stands for  $\Bl^{-1}(\O_{22})$. 
The open subset $p^{-1}(x)\cap \Delta$ is invariant under $St(x)\cong P_2$ (\ref{E:pstab}). The group $P_2$ contains $\Spin(7)$. $\Spin(7)$ has an open orbit (\ref{E:openorb7}) $O_7$ in $\P^7$ and its complement - the quadric. We conclude that $p^{-1}(x)\cap \Delta=O_7$ and $p^{-1}(x)\cap Y=p^{-1}(x)\backslash O_7$ is the quadric $Q$.
\end{proof}

Define  $\L^{\otimes n}$ to be  
\begin{equation}\label{E:ldef}
\L^{\otimes n}=\Bl^*\O(n).
\end{equation}
\begin{proposition}\label{r:direct}
A fiber of \[\E^n=p_*\L^{\otimes n}\] over a point $x\in \OGr(2,11)$ is isomorphic to the space  of $n$-homogeneous functions $\Sym^ns^*_{7}$ on the projective space $\P(s_7)$. The higher direct images $R^ip_*\L^{\otimes n}, n\geq -7$ vanish.
\end{proposition}
\begin{proof}
See \cite{Hartshorne}, Section 5  on cohomology of relative projective spaces.
\end{proof}
\begin{proposition}
The group $\Spin(11)$ acts transitively on $Y$.
\end{proposition}
\begin{proof}
The action of the  stabilizer $St(x)$ of $x\in \OGr(2,11)$  on the quadric $Q_x=p^{-1}(x)\cap Y$ factors through $\Spin(7)\times \widetilde{\GL}(2)$. The group $\Spin(7)$ acts transitively on the quadric $Q_x$. With this remark proposition follows from transitivity of $\Spin(11)$ on $\OGr(2,11)$.

\end{proof}
\begin{remark}\label{r:fiberexact}
The isotropy group $St(x)$ acts through $\Spin(7)\times \widetilde{\GL}(2)$ on the fiber $\E^n_x$. The space $\bigoplus_{n\geq0} \E^n_x\cong \Sym[s^{-1 \ *}(x)]$ (see formula (\ref{E:grading}) for notations).

is a free graded commutative algebra. The data that completely characterizes  the graded bundle $\bigoplus_{n\geq0}\E^n$ is the  representation of $\Spin(7)\times \widetilde{\GL}(2)$ in  the fiber $\E^1_x$, isomorphic to $ s^{-1 \ *}$. 

\end{remark}
\begin{proposition}

Exceptional divisor $Y$ is a fiber bundle over $O_{15}\cong \OGr(5,11)$ with a fiber isomorphic to  $\Gr(2,5)$. The space $Y$ is isomorphic to the space of isotropic flags $\OFl(2,5,11)$.
\end{proposition}
\begin{proof}

  By transitivity of $\Spin(11)$ action  on $O_{15}$ we get $\Bl(Y)=O_{15}$.

By Proposition \ref{P:sing}  and Corollary \ref{C:indeterm}  the scheme $X_{sing}$ coincides with the locus of indeterminacy of the map $p:X\rightarrow \OGr(2,11)$. The exceptional fiber of the blowup of $X$ along $X_{sing}$  in  a chart $\lambda_0\neq 0$  as in Proposition \ref{P:chart} is  $Y=\Proj(\bigoplus_{n\geq 0} I^n/I^{n+1} )$ (see \cite{Hartshorne} about the details of the  blowup construction). The ideal is generated by section $\tau_3$. 
Because of the isomorphism (\ref{E:isoorbits}) $O_{15}$ is equipped with 5-dimensional tautological bundle\[ \{(A,l)|l\in \mbox{isotropic } A\subset \vv^{11}\}\]
From formula (\ref{E:plucker}) we conclude that $Y$ is a subbundle of $\Spin(11)$-homogeneous bundle over $O_{15}$, with a fiber $\P(\Lambda^3A^{'})\cong \P(\Lambda^2A)$.
 Equation (\ref{E:plucker}) defines an embedding of  $\Gr(2,5)$ into  $\P(\Lambda^2A)$.

The fiber of projection $\OFl(2,5,11)\rightarrow \OGr(5,11)$  over a point $\lambda$ is canonically isomorphic to Grassmannian in  $\P(\Lambda^2A_{\lambda})$. From this we conclude that $Y$ is isomorphic to $\OFl(2,5,11)$.
\end{proof}

Let $ \O^{\oplus 11}_{\OGr(2,11)}$ be a trivial bundle over $\OGr(2,11)$ with a fiber isomorphic to $\vv^{11}$. Let $ i:\W\rightarrow \O^{\oplus11}_{\OGr(2,11)}$ be inclusion of the two-dimensional tautological bundle  The trivial bundle $ \O^{\oplus 11}_{\OGr(2,11)}$ is equipped with an inner product induced from $\vv^{11}$. A composition of $i$ with the dual map $i^* \O^{11}_{\OGr(2,11)}\rightarrow  \W^*$ is trivial, because $\W_x$ is  isotropic. We define $\V$ to be $\Ker\ i^*/\W$. The inner product  on $ \O^{11}_{\OGr(2,11)}$ can be pushed to $\V$. This  makes $\V$ a seven-dimensional $\Spin(11)$-homogeneous bundle, equipped with an algebraic orthogonal structure.
Let $V=\V_{x}$ be a fiber over $x\in \OGr(2,11)$.

\begin{proposition}\label{P:pureseven}
An isotropic Grassmannian $\OGr(3,7)\subset \P(\Lambda^3V)$ is isomorphic to a quadric $Q^6\subset \P(s_7)$.
\end{proposition}
\begin{proof}
Let $s_7$ and  $V$ be as in equations (\ref{E:grading}) and  (\ref{E:decomposition1}).
Recall that there is 

 a $\Spin$-equivariant linear isomorphism for $\Sym^2 s_7$ (\ref{E:spin7dec}).
 
The trivial representation corresponds to  equation, that defines  quadric $Q\subset \P^7$ (\ref{E:quadric}). This quadric coincides with the orbit of the highest vector in projectivization of $\Lambda^3V$ under the map $\lambda\rightarrow \gamma(3)(\lambda,\lambda)$. By  uniqueness of the closed orbit, mentioned earlier, this orbit coincides with the Pl\"{u}cker embedding of $\OGr(3,7)$.
\end{proof}

\begin{corollary}
The structure of bundle of quadrics on $Y\cong \OFl(2,5,11)$ coincides the structure defined by the forgetful map $\OFl(2,5,11)\rightarrow \OGr(2,11)$.
\end{corollary}
\begin{proof}
This is a relative version of the previous statement.
\end{proof}

\begin{proposition}\label{E:blowvanish}

Cohomology groups $H^i(\wX,\L^{\otimes n})$ coincide with $H^i(X,\O(n))$.
\end{proposition}
\begin{proof}
The exceptional fiber of $\Bl$ over $x\in \OGr(5,11)$ is isomorphic to $\Gr(2,5)$. By Borel-Weil-Bott theory (see e.g. \cite{Calib}) \[H^i(\Gr(2,5),\O)=\begin{cases}\mathbb{C}&i=0\\\{0\}& i\neq 0 \end{cases}\]
Proposition follows immediately from the spectral sequence for  direct images of $\Bl$.
\end{proof}

\subsection{On the dualizing sheaf of $X$}
Existence of the dualizing sheaf $\omega$ of $X$ is far from obvious, because $X$ is not smooth. Still we shall now that $\omega$ is well defined an isomorphic $ \O_X(-16)$.

We begin this section  reminder of  some basic definitions used in the coherent duality theory.

Let $M$ be a module over commutative ring $R$.
\begin{definition} An element $x \in R$ is called a nonzero divisor on $M$ if $xz = 0$ for $z \in M$ then $z=0$. A sequence $x_1,x_2,...,x_n$ of elements in $R$ is called a regular M-sequence if $x_i$ is a nonzero divisor on $M/(x_1,x_2,...,x_{i?1})M$ for all $i = 1,2,...,n$ and $M/(x_1,x_2,...,x_{n})M\neq 0$. 

\end{definition}

\begin{definition} A local ring $(A,\m)$ is a Cohen-Macaulay ring if there exists a regular sequence $x_1,x_2,...,x_{n}\in \m$ such that the quotient ring is Artinian $A/(x_1,x_2,...,x_{n})$. The maximal $n$ in this case is call the depth.
\end{definition}
An ideal of a commutative ring is said to be {\it irreducible} if it cannot be written as a finite intersection of ideals properly containing it.
\begin{definition} 
A local ring $(A,\m)$ is Gorenstein if for any maximal regular sequence $x_1,x_2,...,x_{n}\in \m$  the ideal $(x_1,x_2,...,x_{n})\subset \m$ is irreducible. 
\end{definition}
The reader may wish to consult \cite{Matsumura} on the details about Cohen-Macaulay (CM) and Gorenstein (G) properties of  rings. In this book the reader can find a generalization of the definitions to nonlocal rings, find that a tensor product  of two CM (G) algebras over the ground field is still CM (G) algebra, see the proof that classes of CM and G algebras are stable under localization. An implication  G$\Rightarrow$CM and the fact that polynomial algebras are Gorenstein are also proved in this book.

It has been shown in \cite{Hochster} that Grassmann varieties have homogeneous coordinate rings which are Cohen-Macaulay. It is also shown that  they are Gorenstein. 

Putting this information together we come to the following proposition.
\begin{proposition}
Local rings of $X$ are normal and Gorenstein.
\end{proposition}
\begin{proof}
Normality of the polynomial algebra is established for example in \cite{ComEisenbud} Proposition 4.10. Relation of normalization and localization is discussed in  \cite{ComEisenbud} Proposition 4.12. 
Proposition follows from the structure of local rings of $X$,  established in Corollary \ref{L:locsmooth} and Proposition \ref{P:chart}. Normality of algebra $B$ (\ref{E:bdef}) has been verified in \cite{normIgusa}.
\end{proof}

We define the Hilbert function of a finitely generated local ring $(A,\m)$ over $\mathbb{C}$ by the formula $A(t)=\sum \dim \m^i/\m^{i+1}t^i$. If $A$ is Cohen-Macaulay we can extract the maximal length  of the regular sequence using the following procedure. The formal power series $A(t)$ is in fact a rational function $q(t)/(1-t)^d$ (see \cite{ComEisenbud}). If numerator and denominator are relatively prime, then the depth is equal to $d$. In the case of algebra $B$ the depth is $7$. It coincides with the dimension of the affine cone over $\Gr(2,5)$.

Let $\Omega^i_Z$ be a sheaf of algebraic differential forms on $Z$ (see \cite{Hartshorne} for details)
When  $Z$ is smooth the dualizing complex $K^{\bullet}$ coincides with $\omega_Z=\Omega^{\dim Z}$. A general theory of coherent duality (see \cite{DualHartshorne}) furnish any scheme with a dualizing complex $K^{\bullet}$. If $Z$ has a  Gorenstein singularity then the dualizing sheaf is quasi-isomorphic to an invertible sheaf (\cite{DualHartshorne} Theorem 9.1). This implies the following.
\begin{proposition}
The dualizing complex $K^{\bullet}$ on $X$ is quasi-isomorphic to some invertible sheaf $\omega$.
\end{proposition}
\begin{proposition}
Let $j:O_{22}\rightarrow X$ we the open embedding. We have an isomorphism
\[\omega\cong j_*j^*\omega.\]
\end{proposition}
\begin{proof}
The group of sections of an abelian sheaf $\F$ on $X$, which
have support in $Z$, is denoted by $\Gamma_Z \F$.
We have a long exact sequence of sheaves (see e.g \cite{DualHartshorne} p.220)
\[ 0 \rightarrow \Gamma_{O_{15}}\omega \rightarrow \omega \rightarrow j_*j^*\omega  \rightarrow R^1\Gamma_{O_{15}}\omega \rightarrow 0\]

Let $\I$ be a sheaf of ideals that defines subscheme $Z$ and let $\F$ be a coherent sheaf. Then  according to \cite{DualHartshorne}
$R^i\Gamma_{Z}\F$ is $\lim_n \cExt_{\O}^i(\O/I^n,\F)$. 

Suppose $Z=O_{15}$. Computations of $\lim_n \cExt_{\O}^i(\O/I^n,\omega)$ can be done locally in the  chart defined in Proposition \ref{P:chart}. The space $\lim_n \cExt_{\O}^i(\O/I^n,\omega)(U)$ equal to some localization of $\mathbb{C}[u,u^{-1},u^{ij}]\otimes \lim_n \Ext_{B}^i(B/I^n,B)$.  By result of   (\cite{DualHartshorne}3.10; \cite{MacdonaldandRYSharp} 2.1) $ \lim_n \Ext_{B}^i(B/I^n,B)=0$ for $0\leq i\leq 6$. 
\end{proof}
\begin{proposition}\label{P:omega}
There is an isomorphism $\omega=\O_X(-16)$.  
\end{proposition}
\begin{proof}
The orbit $O_{22}\subset X$ is a homogenous space of $\Spin(11)$.  The sheaf $j^*\omega$ must be (\cite{DualHartshorne}) a dualizing sheaf on $O_{22}$.  It is isomorphic to $\omega_{O_{22}}$ because $O_{22}$ is smooth. It is a $\Spin(11)$-homogeneous sheaf, and is completely characterized by representation of isotropy group $St(\lambda)$  in the fiber $\omega_{O_{22} \lambda}$. Comparing representations of $\GL(2)$ in $\Lambda^{22}[V+\Ad(\mathfrak{so}_{11})_2+\Ad(\mathfrak{so}_{11})_1]$ (see formula \ref{E:decomp2}) with representation in $\O(1)_{\lambda}$ (this can be readily extracted from decomposition (\ref{E:spinorG2})) we arrive at  the proof.
\end{proof}

The principal statement of the coherent duality theory adapted to our setup takes the form of existence  a canonical non-degenerate pairing \[H^0(X,\O(n))\otimes H^{22}(X,\O(-16-n))\rightarrow \mathbb{C}.\]

\subsection{On the topology  of $X$}
In this section we  compute basic topological invariants of $\wX$-the Poincar\'{e} polynomial of the Chow groups.

Chow groups $CH_k(G/P)$ of complete homogenous spaces of semi-simple groups has been studied in \cite{Chevalley} \cite{BGG}. They are generated by Schubert cells and have no torsion. The map to singular homology $CH_k(G/P)\rightarrow H_{2k}(G/P)$ is an isomorphism. Hodge cohomology $H^{ij}=H^{i}(G/P,\Omega^j )$ is in nonzero iff $i=j$. From de Rham isomorphism we have $\rank CH_k(G/P)=\dim(H^{p,p})$, $p+n=\dim_{\mathbb{C}}G/P$.

Our plan is to take advantage of fiber-bundle structure on $\wX$ and use Leray-Hirsh arguments for homology computation. For this we need to know homology of the base of the fibration. We study this in the following proposition.
\begin{proposition}\label{P:topogr} \ 
\\
\begin{enumerate}
\item $\dim_{\mathbb{C}}\OGr(2,11)=15$.
\item The index $\ind$ is defined as $\O(-\ind) =\omega_{\OGr(2,11)}$, where $\O(-1)$ is the dual of the  line bundle, corresponding to the Pl\"ucker embedding.

Then $\ind \OGr(2,11)=8$.

\item Let $H_{\OGr(2,11)}(t)$ be the Poincar\'{e} series of the Hodge cohomology \\ $\sum_{ij=0}^{15}\dim H^{p,p}(\OGr(2,11))t^p$. Then 
\[\begin{split}&H_{\OGr(2,11)}(t)=1+t+2\,{t}^{2}+2\,{t}^{3}+3\,{t}^{4}+3\,{t}^{5}+4\,{t}^{6}+4\,{t}^{7}\\&+
4\,{t}^{8}+4\,{t}^{9}+3\,{t}^{10}+3\,{t}^{11}+2\,{t}^{12}+2\,{t}^{13}+
{t}^{14}+{t}^{15}.
\end{split}\]
\end{enumerate}
\end{proposition}
\begin{proof} \ 
\\
\begin{enumerate}
\item $\dim_{\mathbb{C}}\OGr(2,11)=\dim_{\mathbb{C}}\Spin(11)-\dim_{\mathbb{C}}P_2=55-40=15$.
\item $\det(\Lambda^2U+V\otimes U)\overset{\ddef}{=}\Lambda^2U^{\otimes \ind}$

\item The algebra of cohomology $H^{pp}(\OGr(2,11))$ is generated by two classes $c_1\in H^{1,1}$ and $c_2\in H^{2,2}$ degree two and four. Defining relations in the algebra have degrees $16$ and $20$ (see.e.g. \cite{Tamvakis}). Thus the degree seven Maclaurin  polynomials of $H(t)$ and of the function $\frac{1}{(1-t)(1-t^2)}$ coincide. The  formula for $H(t)$ then follows from Poincar\'{e} duality in  $H^{pp}(\OGr(2,11))$.
\end{enumerate}
\end{proof}

\begin{proposition}
 The map $CH_*(\wX)\rightarrow H_*(\wX,\mathbb{Z})$ is an isomorphism. Moreover 
\[\begin{split}
&H_{\wX}(t)={t}^{22}+2\,{t}^{21}+4\,{t}^{20}+6\,{t}^{19}+9\,{t}^{18}+12\,{t}^{17}+
16\,{t}^{16}+20\,{t}^{15}+23\,{t}^{14}+26\,{t}^{13}+27\,{t}^{12}\\
&+28\,{t}^{11}\\
&+27\,{t}^{10}+26\,{t}^{9}+23\,{t}^{8}+20\,{t}^{7}+16\,{t}^{6}+
12\,{t}^{5}+9\,{t}^{4}+6\,{t}^{3}+4\,{t}^{2}+2\,t+1
\end{split}
\]
\end{proposition}
\begin{proof}

The algebra $CH^*(\wX)$ 

 is a free $CH^*(\OGr(2,11))$-module, because $\wX$ is a projective bundle over $\OGr(2,11)$ (see \cite{inersectionFulton} Theorem 3.3). The same theorem holds for singular cohomology (see \cite{Hatcher} Leray-Hirsh theorem). From this we deduce that $CH_*(\wX)$ is isomorphic to $H_*(\wX,\mathbb{Z})$.  Let $\L$ (\ref{E:ldef}) be the relative tautological bundle over the relative projective space $p:\wX\rightarrow \OGr(2,11)$. Then the free basis of $H(\wX)$ is constituted by $1, c_1,\dots, c_1^7(\L)$. Thus $Q(t)=H_{\OGr(2,11)}(t)(1+\cdots+t^7)$, with $H_{\OGr(2,11)}(t)$ defined in Proposition \ref{P:topogr} .
\end{proof}

Locally factorial varieties (for example regular or smooth varieties) satisfy $\Pic(Y)=CH^1(Y)$ (see \cite{ComEisenbud} p. 260).

This  implies that $\Pic(\wX)=CH^1(\wX)=H^2(\wX,\mathbb{Z})=\mathbb{Z}^2$. It is generated by pullback  $p^*\O_{\OGr(2,11)}(1)$ and $\Bl^* \O(1)_{X}$ - the tautological line bundle of $\P(\E)$
(\ref{E:vectprb}).

\subsection{More on singularities of $X$}
We shall establish momentarily   that algebraic variety $X$ has canonical singularities (the concept introduced by  Reid in the work on the minimal model program)
\begin{proposition} \label{P:canonical}
Let $\omega_X, \omega_{\wX}$ be the dualizing sheaves on $X$ and respectively on $\wX$. Then 
\begin{equation}\label{e:pullback}
\omega_{\wX} =\Bl^*\omega_X+4Y
\end{equation}
where $Y\subset \wX$ is the exceptional divisor.
\end{proposition}
\begin{proof}
According to Proposition \ref{P:omega} $\omega_X=\O_X(-16)$, $\omega_{\wX}=p^*\omega_{\OGr(2,11)}\otimes \omega_{\wX/\OGr(2,11)}$, where $\omega_{\wX/\OGr(2,11)}$ is the relative canonical class. By Proposition \ref{P:topogr} $\omega_{\OGr(2,11)}=\O_{\OGr(2,11)}(-8)$.

 In order to prove formula (\ref{e:pullback}), due to the absence of torsion in $CH_*$ it is suffice to check it numerically with respect to some linearly independent set of one-cycles. 
For these we choose in $X$ two curves $\Sigma_i, i=1,2$ isomorphic to $\P^1$s.  $\Sigma_1$ is  a projectivization of a linear space $v\otimes W_l\subset s_{11}$, spanned by vectors $v\otimes e_1, v\otimes e_2$ as in Proposition \ref{E:vanish}. We choose $v\in s_7$ such that $(v,v)\neq 0$ (see formula \ref{E:quadric}). The last condition guarantees that $\Sigma_1\subset O_{22}$. $\Sigma_2$ is  a projectivization of a linear space spanned by $v_1\otimes e, v_2\otimes e\in  s_{7}\otimes W_l$. It is the $\Bl$-image  of a curve $\widetilde{\Sigma}_2$ that contains in the fiber $p$. 
 The cycles defined by $\Sigma_i$ are linearly independent, because $p_*(\Sigma_2)=0$ in $CH_1(\OGr(2,11))$.
 
 We have the following pairings  $\<\Bl^*\omega,\Sigma_i\>=-16$, \[\<\omega_{\wX},\Sigma_1>= \<\omega_{\OGr(2,11)},p(\Sigma_1)\>=\<\O_{\OGr(2,11)}(-8),p(\Sigma_1)\>=-8,\] \[\<\omega_{\wX},\Sigma_2\>=\<\omega_{\wX/\OGr(2,11)},\Sigma_2\>=\<\omega_{\P^7},\Sigma_2\>=-8,\] $\<Y,\Sigma_1\>=0, \<Y,\Sigma_2\>=2.$ The last equality holds because intersection of $Y$ with a $p$ fiber is a quadric (Corollary \ref {C:total}). With  preparatory work behind, the formula $\<\omega_{\wX},\Sigma_i\> =\<\Bl^*\omega_X,\Sigma_i\>+4\<Y,\Sigma_i\>$ $i=1,2$ is  obvious.
 \end{proof}
\begin{remark}
It follows from Proposition \ref{P:canonical} that variety $X$ is an example of a Fano manifold with  canonical singularities (see e.g. \cite{parshin} for details).  

\end{remark}

\section{Applications to eleven-dimensional supergravity}\label{S:applications}

In this section we shall describe an alternative formulation of eleven-dimensional supergravity announced in the abstract. We shall arrive to this formulation through a series of a quasi-isomorphisms. We shall describe our construction  for polynomial fields. Though this is not physically very realistic assumption, it lets to  simplify the  statements. In the end we discuss how to work with analytic fields.

The algebra $Gr^{\infty}$ (\ref{E:gr}) contains a subalgebra $Gr^{pol}=A\otimes \Lambda[s^*_{11}]\otimes \Sym[\vv^{11}]$. One of the advantages of working with $Gr^{pol}$ is that it has a $\mathbb{Z}$ grading, compatible with the action of $D$. This contrasts with $Gr^{\infty}, Gr^{an}$ modification, for which  only $\mathbb{Z}_2$-grading is possible. By definition $\deg(\xi^{\alpha})=1, \deg(\lambda^{\alpha})=\deg(x^i)=2$. The algebra $Gr=Gr^{pol}$  admits various reformulations.

\subsection{Sheafification of $Gr$}

Let us introduce a sheaf of graded  algebras  $\A=\bigoplus_{n\geq 0}\O(n)$ defined over $X$. We use it to define a sheaf of differential graded algebras \[\cGr=(\A\otimes \Lambda[\xi^{1},\dots,\xi^{32}]\otimes \Sym[\vv^{11}] ,D).\]
Vanishing results for cohomology of $\O(n)$ (Proposition \ref{P:vanishing}) lets us to prove the following statement
\begin{proposition}\label{P:comp11}
Hypercohomology of $\cGr$ coincide with $H(Gr)$.
\end{proposition}
From this and from Proposition \ref{E:blowvanish} we deduce a corollary.
\begin{corollary}\label{C:equiv}
Hypercohomology of $\Bl^*\cGr$ coincide with $H(Gr)$.
\end{corollary}
We define a sheaf of differential graded algebras 
\begin{equation}\label{E:grdef}
\H=p_*\Bl^*\cGr
\end{equation}
 on $\OGr(2,11)$.
\begin{corollary}\label{C:quasiHGr}
The hypercohomology of $\H$ coincide with $H(Gr)$.
\end{corollary}
\begin{proof}
This follows from Corollary \ref{C:equiv}, Proposition \ref{r:direct} and spectral sequence for direct images of a morphism.
\end{proof}

\subsection{A homogeneous space of super-Poincar\'{e} group}\label{S:homodef}
In this section we define a homogeneous space $L$ of  super-Poincar\'{e} group, that 

has some significance in eleven-dimensional supergravity.

Let \[\so_{11}\ltimes \susy\]be the complexified super-Poincare Lie algebra.
As usual $\susy$ is a direct sum of  $\vv^{11}$ (even part), $s_{11}$ (odd part). The only nontrivial bracket is defined by the formula $[s,s']=\Gamma(s,s'),s,s'\in s_{11}$.
As far as global structure of the corresponding group is concerned then  according to (\cite{Bernstein} 2.10. Super Lie groups) a supergroup $G$, i.e. a super-manifold, equipped with a group structure is completely determined by the Lie algebra $\mathfrak{g}$ of left-invariant vector fields and the topology of  purely even subgroup $G_0$. In context of super-Poincar\'e group  we require that underlying even group to be $\Spin(11)\ltimes \vv^{11}$.

Citing the same source we  claim that a homogenous space $X=G/H$ of a supergroup $G$ and isotropy subgroup $H$ is completely determined by Lie algebras $\g,\h$ and supporting manifold $G_0/H_0$.

We define $L$ to be a quotient of super-Poincar\'{e} group $\Spin(11)\ltimes \susy$ by an isotropy subgroup $P_2\ltimes \Pi \t$ . ( $P_2$ as in (\ref{E:pstab}), $\t=s^{-1}$ is as in formula (\ref{E:grading})).
\begin{remark}
By Proposition \ref{E:vanish} $\Pi \t$ is an abelian subalgebra in $\susy$.
\end{remark}

Super-Poinca\'e Lie algebra also has its own $\mathbb{Z}$ grading, The $\so_{11}$-part has zero grading, $\Pi s_{11}$ has grading minus one (opposite to $\deg(\xi^{\alpha})$), translations have grading minus two (opposite to $\deg(x^i)$).
\begin{remark}
The grading on $\so_{11}\ltimes \susy$ induces a $\mathbb{C}^*$ action on the manifold $L$. The fixed points of the action are $\OGr(2,11)\subset L$. $\mathbb{C}^*$ acts on cohomology of any equivariant coherent sheaf on $L$. In particular  $H^i(L,\O)$ splits into a direct sum \[H^i(L,\O)=\bigoplus_{k}H^{i,k}(L,\O)\] of $\mathbb{C}^*$-weight spaces. It will be convenient for us to modify cohomological grading \[H_{tot}^k(L,\O)=\bigoplus_{i+j=k}H^{i,j}(L,\O)\]
\end{remark}

\subsection{Algebra of linearized supergravity}

\begin{proposition}
There is an isomorphism of cohomology $H^k(Gr^{pol})$ and 
and cohomology of the structure sheaf $H_{tot}^k(L,\O)$. 
\end{proposition}
\begin{proof}
We know that cohomology $Gr$ coincide with hypercohomology of $\H$ (Corollary \ref{C:quasiHGr}).
The fiber of $\H$ at the point $x\in \OGr(2,11)$ is isomorphic to $H=\Sym[\t^*]\otimes \Lambda[s^*_{11}]\otimes \Sym[\vv^{11}]$. 
Homogeneous space $\OGr(2,11)$ is trivially a $\Spin(11)\ltimes \susy$ space.
We have an equivariant fibration
$L\rightarrow \OGr(2,11)$, corresponding to inclusion of isotropy subalgebras $\p_2\ltimes \Pi \t \subset \p_2\ltimes \susy$. Let $SUSY$ and $\Pi T$ be the super group schemes corresponding to Lie algebras $\susy$ and $\Pi\t$.
A fiber of the projection  is isomorphic to $SUSY/\Pi T$.
The algebra of polynomial functions on the fiber coincides with some subalgebra of functions on $SUSY$. The subalgebra consists of elements, invariant with respect to right translations on $\Pi T$. The algebra of global functions $\O(SUSY)$ is isomorphic to $\Lambda[s^*_{11}]\otimes \Sym[\vv^{11}]$, because $\susy$ is a nilpotent Lie algebra. Left invariant vector fields $\eta_{\alpha}\in \Pi s_{11}\subset \susy$ act by the formula $\frac{\sd \ }{\sd \theta^{\alpha}}-\Gamma^i_{\alpha\beta}\lambda^{\alpha}\theta^{\beta}\frac{\sd}{\sd x^i}$ (see e.g. \cite{DF} formula 1.18). The subalgebra of invariants is the zero cohomology of the Cartan-Chevalley complex (see e.g. \cite{Fuchs} on cohomology of Lie algebras) $\Sym[\t^{*}]\otimes \Lambda[s^*_{11}]\otimes \Sym[\vv^{11}]=C(\Pi\t,\Lambda[s^*_{11}]\otimes \Sym[\vv^{11}])$. The module $\Lambda[s^*_{11}]\otimes \Sym[\vv^{11}]$ is co-induced. By Shapiro lemma its higher cohomology vanish. 
We see that $\O(SUSY/\Pi T)$ is quasi-isomorphic to $\Sym[\t^{*}]\otimes \Lambda[s^*_{11}]\otimes \Sym[\vv^{11}]$, which isomorphic to $H$. 
The fibers of projection $L\rightarrow \OGr(2,11)$ are affine. The proposition follows from the spectral sequence for cohomology of fibrations.

\end{proof}

Reference to Shapiro lemma in the proof of the last proposition can be replaced by explicit computation.
The computation uses decomposition (\ref{E:grading}).  Let $t^k,t^{'k} u^{ik}, k=1,\dots,8, i=1,2$ be linear coordinates in spaces $\t=s^{-1}, s^{1}$ and $s^0$; $\tau^k,\tau^{'k},\nu^{ik}$ -  coordinates in  $\Pi s^{-1}, \Pi s^{1}$ and $\Pi s^0$.  

Under the map  \[Gr^{\infty}\rightarrow H^{\infty}=\Sym[\t^*]\otimes \Lambda[s^*_{11}]\otimes C^{\infty}(\mathbb{R}^{11})\]   the image of  $e^i=\Gamma^i_{\alpha\beta}\lambda^{\alpha}\theta^{\beta}\in Gr^{\infty}$ 
can be written as a linear combination of 
\[\tilde{e}^p=\gamma^p_{kl}t^k\tau^{'l}\quad p=1,\dots,7,\]
\[\tilde{e}^i=t^{k}\nu^{ki}, i=1,2.\]
We use  $\Spin(7)$ notations and  formulas (\ref{E:goct1},\ref{E:goct2}).

 We have  $\tilde{e}^p=D(\gamma^p_{kl}\tau^k\tau^{'l})$, $\tilde{e}^i=D(\tau^{k}\nu^{ki})$ in  $H^{\infty}$. Introduce notations $f^p=\gamma^p_{kl}\tau^k\tau^{'l}$, $g^i=\tau^{k}\nu^{ki}$. We choose  coordinates  $x^p$ on the space $\vv^7$, $u^i$ on $U$, $u_i$ on $U^{'}$ in decomposition (\ref{E:decomposition1}) to 
  satisfy $Dx^p=\tilde{e}^p$ , $Du^i={e}^i$ and  $Du_i=0$ . By construction elements 
  \begin{equation}\label{E:superfunct}
  y^p=x^p-f^p, h^i=g^i-u^i,  u_i, \tau^{'k}, \nu^{il}
  \end{equation} are $D$-cocycles in $H^{\infty}$. We use (\ref{E:superfunct}) augmented by $\tau^k$ as coordinates on $\vv^{11}\times \Pi s_{11}$. In these coordinates the complex $H^{\infty}$ becomes a tensor product of  odd De Rham complex $\Omega[\Pi \t]$ and the algebra $F^{\infty}$ of $C^{\infty}$ functions in  (\ref{E:superfunct}). The cohomology $H^{\infty}$ can be computed with the aid of  the K\"{u}neth formula and coincide with $F$.
  
  These arguments let to compute cohomology $H^{an, \mathbb{R}}=\Sym[\t^*]\otimes \Lambda[s^*_{11}]\otimes \O^{an}(\mathbb{R}^{11})$ and $H^{an}=\Sym[\t^*]\otimes \Lambda[s^*_{11}]\otimes \O^{an}(\vv^{11})$.
Dolbeault theory on complex super-manifolds has been treated in \cite{Haske}. In particular Theorem 3.4 in \cite{Haske} shows equivalence of Dolbeault and \v{C}ech approaches to cohomology of coherent sheaves.
\begin{proposition}
There is an isomorphism of cohomology $H^k(Gr^{an})$ and 
and cohomology of the structure sheaf $H_{tot}^k(L,\O^{an})$ $k\in \mathbb{Z}_2$. 
\end{proposition}
\begin{proof}
We follows the lines of algebraic proof.
The key moment is to use GAGA (\cite{Serre} Section 12 Theorem 1 ) for comparison analytic and algebraic cohomology of compact (complete) spaces that appear in the proof of Proposition \ref{P:comp11}, Corollaries  \ref{C:equiv}, \ref{C:quasiHGr}.  Finally we use Stein property of $\vv^{11}$ (see \cite{Hoermander} Definition 5.1.3, Theorem 2.5.5, Theorem 7.4.1.) to prove that analytic cohomology of a  fiber of projection $L\rightarrow \OGr(2,11)$ coincide with $F_{an}$.
\end{proof}

\begin{corollary}
There is one-to one correspondence between equivalence classes of solutions $\dbar f=0$  $f\in \bigoplus _{p\geq 0}\Omega^{0p}_{L}$ and cohomology classes in $Gr^{an}$.
\end{corollary}
\bigskip

{\bf\Large  Appendix}
 
\appendix

\section{Homological properties of Serre algebra of $X$}\label{S:homology}
Most of the statements presented in this appendix are well known to experts in supergravity. The main motivation for incorporating this material into the paper is to complement these facts (or conjectures in mathematical language)  by proofs. Unfortunately the situation with the proofs is far from being satisfactory, because they rely on the use of algebra systems  {\it Mathematica}, {\it  Macauley2} and LiE. The main text uses result of Corollary  \ref{C:degree}.

The space $\vv^{11}$ has a physical  interpretation of a complexified space-time.
The number of nonzero entries in matrices $\Gamma^i_{\alpha\beta}$ for a given $i$ grows exponentially with dimension of the spaces-time.  Computations with $\Gamma$s are hard to handle by hands. This is why  physicists designed a {\it Mathematica} package {\verb gamma.m }  \cite{gamma} to generate $\Gamma^i_{\alpha\beta}$. The following code in  {\it Mathematica}  creates an array of  $v^i$:
\begin{verbatim}
<< gamma.m
Var = Array[x, {32}];
Vec := Array[vec, {11}];
v[k_] := Simplify[
Sum[Sum[(WeylGamma[11, k].WeylC1[11])[[i]][[j]]*Var[[i]]*Var[[j]], 
{i, 1, j}], {j, 1, 32}]]
\end{verbatim}

One of the natural questions that can be asked about $A$ is the following. The algebra $A$ is a graded $G=\mathbb{C}[\lambda^1,\dots,\lambda^{32}]$-module. Find the minimal free resolution. This is a formidable task, because of the large homological dimension of $G$, but {\it Macauley2} \cite{Macaulay2Doc} is fit  for the   job. Here is a suitable {\it  Macauley2} code
\begin{verbatim}

i1 : G=QQ[x_(1) .. x_(32)];
i2 : a = matrix {{...}};
             1       11
o2 : Matrix G  <--- G
i3 : A= coker a
o3 = cokernel |....|
o3 : G-module, quotient of G
i4 : R=res A
\end{verbatim}

\ \\
\begin{verbatim}
      1    11   66   263  352  352  263  66   11   1
o4 = G<-- G<-- G<-- G<-- G<-- G<-- G<-- G<-- G<-- G<-- 0
     0    1    2    3    4    5    6    7    8    9    10
\end{verbatim}
The dots in {\verb matrix \{\{...\}\}  } stand for the array of relations found with {\verb gamma.m }. Command {\verb R_i } gives the degree of the generators of the $i$-th module of syzygies.
We conclude that with the degrees taken into account the minimal resolution has the form
\begin{equation}\label{E:resolution0}\begin{split}
&R=G^1(0)\leftarrow 
G^{11}(-2) \leftarrow 
G^{66}(-4) \leftarrow 
G^{32}(-5) +G^{231}(-6) \leftarrow 
G^{352}(-7) \leftarrow \\
& \leftarrow  G^{352}(-9) \leftarrow 
G^{32}(-11) +G^{231}(-10)  \leftarrow 
G^{66}(-12) \leftarrow 
G^{11}(-14) \leftarrow 
G^{1}(-16) \leftarrow 0\end{split}
\end{equation}
In $G(n)$ $n$ stands for the grading shift.
It takes about 30 minutes to pass all steps of  La Scala's algorithm on 2.13GHz Processor with 4GB internal memory. The book \cite{CompMacaulay2} and the manual \cite{Macaulay2Doc} will give a references on the internal structure of the algorithms employed.

The resolution (\ref{E:resolution0}) allows to compute the graded groups  \[\Tor^G_i(A,\mathbb{C})=\bigoplus_j \Tor^G_{ij}(A,\mathbb{C}).\] It makes sense to define a generating function of dimensions  $\Tor^G_i(A,\mathbb{C})(t)$ and the Euler characteristic $\chi(\Tor)(t)=\sum_{i\geq 0} (-1)^i\Tor^G_i(A,\mathbb{C})(t)$. In general  Euler characteristic gives a limited information about cohomology of a complex. The statement has an exception when all cohomology but one vanish. This happens $j$ component-wise in $\Tor$ groups at hand. 

There is an alternative way to compute $\Tor^G_i(A,\mathbb{C})$. It is to resolve the second argument $\mathbb{C}$. The classical Koszul resolution $G\otimes \Lambda[\xi^{1},\dots,\xi^{32}]$ with differential $\lambda^{\alpha}\sd_{\xi^{\alpha}}$ written in super notations  does the job.
The complex 
\begin{equation}\label{E:koszul}
K=(A\otimes \Lambda[\xi^{1},\dots,\xi^{32}] ,\lambda^{\alpha}\sd_{\xi^{\alpha}})
\end{equation} still computes $\Tor^G(A,\mathbb{C})$. The group  $\Spin(11)$ is the group of symmetries of $G$, $A$ and $A\otimes \Lambda[\xi^{1},\dots,\xi^{32}]$. 
\begin{lemma}
The differentials in the minimal resolution $R$   (\ref{E:resolution0}) commute  with $\Spin(11)$ action.
\end{lemma}
\begin{proof}
Left to the reader.
\end{proof}
%\begin{proof}
%A free resolution is unique up to a homotopy \cite{McL} . A minimal free resolution of a finitely generated graded module is unique up to a grading preserving isomorphism.  We infer that the twists of $R$ by elements of $\Spin(11)$ are isomorphic and some extension of $\Spin(11)$ by  $Aut_G(A)\cong \mathbb{C}^*$ acts on $R$.

% The extension is trivial because $\Spin(11)$ is simple and simply-connected. 
%\end{proof}

An irreducible representation of a semi-simple group is labeled by coordinates of the highest weight (see e.g. \cite{FultonRep} for details). A representation of $\Spin(11)$ has five coordinates $(w_1,w_2,w_3,w_4,w_5)$. For example an eleven-dimensional fundamental representation has coordinates $(10000)$, spinor representation - $(00001)$.

The differential in $K$ decreases degree in $\xi^{\alpha}$ . For this reason we use homological notations.

\begin{proposition}(cf. \cite{CNT})
Let $H_i(K)=\bigoplus_{j} H_{ij}(K)$ be the decomposition into graded components. We have the following isomorphisms:
\begin{equation}
\begin{split}
&H_{0,0}=(00000)\\
&H_{1,2}=(10000)\\
&H_{2,4}=(01000)+(10000)\\
&H_{3,5}=(00001),H_{3,6}=(00000)+(00100)+(20000)\\
&H_{4,7}=(00001)+(10001)\\
&H_{5,9}=(00001)+(10001)\\
&H_{6,10}=(00000)+(00100)+(20000),H_{6,11}=(00001)\\
&H_{7,12}=(01000)+(10000)\\
&H_{8,14}=(10000)\\
&H_{9,16}=(00000)
\end{split}
\end{equation}
All other homology groups are trivial.
\end{proposition}
\begin{proof}
The groups $\Tor^G(A,\mathbb{C})$ are $\Spin(11)$-representations. It makes sense to define Euler characteristic   $\chi(\Tor)(t)$ with values in the ring of virtual finite-dimensional representations $Rep(\Spin(11))$. Such Euler characteristics has been computed in \cite{CNT}. As it was argued above the structure of $R$ lets to unambiguously interpret  the coefficients of $\chi(\Tor)(t)\in Rep(\Spin(11))[[t]]$ as groups $H_{ij}(K)$.
\end{proof}
\begin{remark}\label{R:selfdual}
We know that  the differentials in the complex $R$ commutes with $\Spin(11)$-action. Representation theory uniquely fixes the  differentials. This has been used in \cite{BerNek} (who worked under assumption of  acyclicity of $R$)to find the formulas for the boundary maps in terms $\Gamma$ matrices (\cite{BerNek}  formula (5.5)). By inspection of the formulas we see that the complex $R$ is self-dual:
\[\Hom_G(R,G)\cong R.\]
\end{remark}

Resolution (\ref {E:resolution0}) enables us to compute 

the Hilbert polynomial $H(n)=\dim A_n, n\gg 0$ of the projective  $X$:
\[H(n)= \frac {\left( 5\,{n}^{4}+160\,{n}^{3}+2107\,{n}^{2}+13232\,n+38760
 \right)}{5\times 2\times3 \times7 \times19!} \prod_{i=7}^9(n+i) \prod_{i=1}^{15}(n+i)  
\]
\begin{corollary}\label{C:degree}
Degree $\deg(H)$ is equal to twenty two.

From this we infer (see e.g. \cite{ComEisenbud}, \cite{Hartshorne} for discussion of Hilbert polynomials) that the greatest dimension of irreducible components of $\Proj(A)$ is $22$. 
\end{corollary}
This agrees with a computation done in \cite{BerNek}. Let $\O(n)=\O_X(n)$ be the tautological bundle over $X$.
 The number  $H(n)$ coincides with the Euler characteristic  $\chi(\O(n))$( see e.g. \cite{Hartshorne}(Ex 5.2)). The formula for $H(n)$ tells us that $\chi(\O(n))=0$ for $n=-1, \dots, -15$. The next Proposition is a refinement of this observation.
 \begin{proposition}\label{P:vanishing}
 \begin{enumerate}
 \item The sheaves $\O(n), n=-1, \dots, \O(-15)$ are acyclic, i.e. $H^l(X,\O(n))=0, l=0,\dots 22$.
 \item $H^l(X,\O(n))=0, l=1,\dots 22, n\geq 0$.
 \item $H^l(X,\O(n))=0, l=0,\dots 21, n\leq -16$
 \end{enumerate}
 \end{proposition}
 \begin{proof}
 Let $i$ be the closed embedding of $X$ into $\P^{31}$. Sheafification of $R$ gives a resolution $\calR$ of $i_{*}\O_X$ by locally free sheaves on $\P^{31}$ of the form $\calR^i=\bigoplus_k \O_{\P^{31}}(n_{k,i})$. The sheaves $\O_{\P^{31}}(n)$ are acyclic for $-31\leq n\leq -1$. They have trivial higher cohomology for $n\geq 0$ and have only top degree cohomology if $n\leq 32$(see e.g. \cite{Hartshorne} Theorem 5.1). By Serre duality $H^{31}(\P^{31},\calR^i(n))$ is $\mathbb{C}$-dual to the suitable subgroup of $\Hom_G(R,G)$. The statement readily follows from the hypercohomology spectral sequence for $\calR(n)$ and self-duality of $R$ (Remark \ref{R:selfdual}). 
 \end{proof}
 
 The complex $K$ is a differential graded  algebra.  Let $l:H_{9,16} \rightarrow \mathbb{C}$ be an isomorphism of vector spaces.
 \begin{proposition}
 There is a perfect pairing $(a,b)$ in cohomology $H(K)$, defined by the formula \[(a,b)=l(ab)\]
 \end{proposition}
 \begin{proof}
 The proof repeats the arguments of \cite{MSch} (Theorem 55). The reader should keep in mind that though we haven't established smoothness of $X$ (in fact $X$ is not smooth) we do have the necessary cohomology vanishing results for $\O_X(n)$ (Proposition \ref{P:vanishing}) that suffice for the proof.
 \end{proof}

\section{Proof of Proposition \ref{P:subs}}\label{A:subs}
After elimination of $r^1,r^2$ the pair of equations in question becomes
\[(w^2v^1- w^1v^2,u^1)=0\quad (w^2v^1- w^1v^2,u^2)=0\]
Then 
\[\begin{split}& (w^1v^2-w^2v^1,u^1)=\frac{1}{(v^2,v^2)+w^1w^1}\times\\
&\left(-w^1w^1(u^1u^2,u^1)-w^1(v^1,u^2)(u^1,u^1)+w^1(v^1,u^1)(u^2,u^1)\right.\\
&+w^1(v^1.(u^1.u^2),u^1)+(v^1,u^2)(v^1u^1,u^1)-(v^1,u^1)(v^1u^2,u^1)+w^1(u^1,u^2)(v^1,u^1)\\
&\left.-(v^1,u^1u^2)(v^1,u^1)\right).
\end{split}\]
We use (\ref{E:oidentity1}) to eliminate $(u^1u^2,u^1), \frac{(v^1,u^1)}{w^1}(v^1u^1,u^1)$ . Because of the identity \[(ab,cd)=2(a,c)(b,d)-(ad,cb)\] (see \cite{ConwaySmith}) we  replace $(v^1.(u^1.u^2),u^1)=-(u^1u^2,v^1u^1)$ by $-2(u^1,v^1)(u^2,u^1)+(u^1u^1,v^1u^2)=-2(u^1,v^1)(u^2,u^1)+(u^1,u^1)(v^1u^2)$
\[\begin{split}& (w^1v^2-w^2v^1,u^1)=\frac{1}{(v^2,v^2)+w^1w^1}\times\\
&\left(-w^1(v^1,u^2)(u^1,u^1)+w^1(v^1,u^1)(u^1,u^2) -2w^1(u^1,v^1)(u^1,u^2)\right.\\
&+w^1(u^1,u^1)(v^1,u^2)+(v^1,u^1)(v^1,u^1u^2)+w^1(u^1,u^2)(v^1,u^1)\\
&\left.-(v^1,u^1u^2)(v^1,u^1)\right)=0
\end{split}\]
Likewise
\[\begin{split}& (w^1v^2-w^2v^1,u^2)=\frac{w^1}{(v^2,v^2)+w^1w^1}\times\\
&\left(-w^1w^1(u^1u^2,u^2)-w^1(v^1,u^2)(u^1,u^2)+w^1(v^1,u^1)(u^2,u^2)\right.\\
&+w^1(v^1.(u^1.u^2),u^2)+(v^1,u^2)(v^1u^1,u^2)-(v^1,u^1)(v^1u^2,u^2)+w^1(u^1,u^2)(v^1,u^2)\\
&\left.-(v^1,u^1u^2)(v^1,u^2)\right)
\end{split}\]
We replace $(v^1.(u^1.u^2),u^2)=-(u^1u^2,v^1u^2)$ by $-(v^1,u^1)(u^2,u^2)$. We get

\[ \frac{-w^1(u^2,u^2)(u^1,v^1)+w^1(v^1,u^1)(u^2,u^2)}{(v^2,v^2)+w^1w^1}=0\]

%\bibliographystyle{plain.bst}
%\bibliography{gr11}
\end{document}